%% file: main.tex
\documentclass[orivec, envcountsame, runningheads]{llncs}

\usepackage{amsmath}
\usepackage{amssymb}
\usepackage{stmaryrd} 
  \SetSymbolFont{stmry}{bold}{U}{stmry}{m}{n} 
\usepackage{multirow}
\usepackage{threeparttable}

\usepackage{xspace} 
\usepackage{cite}
\usepackage{rotating}
\usepackage{listings}
\usepackage{newfloat}
  \DeclareFloatingEnvironment[within=none,fileext=lop]{listing}[Listing]

\usepackage{combelow} 

\usepackage{paralist}  

\usepackage{subfigure}

\usepackage[kerning=true]{microtype} 

\usepackage{tikz}
\usetikzlibrary{calc,shapes,arrows}

\input{prelude}

\title{Runtime Verification of Temporal Properties over Out-of-order
  Data Streams}
\titlerunning{Runtime Verification over Out-of-order Data Streams}

\author{David Basin\inst{1} \and Felix Klaedtke\inst{2} \and 
  Eugen Z\u{a}linescu\inst{3}}
\authorrunning{D. Basin et al.}

\institute{ETH Zurich, Department of Computer Science \and 
  NEC Laboratories Europe, Heidelberg \and 
  Technische Universit\"at M\"unchen}

\allowdisplaybreaks

\begin{document}

\maketitle

\begin{abstract}
  We present a monitoring approach for verifying systems at runtime.
  Our approach targets systems whose components communicate with the
  monitors over unreliable channels, where messages can be delayed or
  lost.  In contrast to prior works, whose property specification
  languages are limited to propositional temporal logics, our approach
  handles an extension of the real-time logic MTL with freeze
  quantifiers for reasoning about data values.  We present its
  underlying theory based on a new three-valued semantics that is well
  suited to soundly and completely reason online about event streams
  in the presence of message delay or loss.  We also evaluate our
  approach experimentally.  Our prototype implementation processes
  hundreds of events per second in settings where messages are
  received out of order.
\end{abstract}

\input{intro}
\input{prelim}
\input{mtl}      
\input{new_monitor}

\input{casestudy}
\input{related}
\input{concl}

\bibliographystyle{abbrv}
\bibliography{references}

\newpage
\appendix
\input{app_mtl}
\input{app_pseudocode}

\input{app_proof}
\input{app_casestudy}

\end{document}

%% file: prelude.tex
\newcommand{\ignore}[1]{}

\newcounter{PolicyCnt}
  {\end{equation}\ignorespacesafterend}

\renewcommand{\phi}{\varphi}
\renewcommand{\rho}{\varrho}

\newcommand{\set}[1]{\{#1\}}
\newcommand{\setx}[2]{\{#1\mathbin{\,|\,}#2\}}

\newcommand{\Nat}{\mathbb{N}}
\newcommand{\N}{\mathbb{N}}
\newcommand{\Q}{\mathbb{Q}}
\newcommand{\Qpos}{\mathbb{Q}_{\geq 0}}

\newcommand{\pto}{\nrightarrow}  
\newcommand{\pdef}{\operatorname{def}}

\newcommand{\MTL}{\textup{MTL}}
\newcommand{\MTLdata}{{\MTL$^{\downarrow}$}\xspace}

\newcommand{\timestamp}{\operatorname{ts}}
\newcommand{\position}{\mathit{pos}}
\newcommand{\istp}{\mathrm{tp}}
\newcommand{\tc}{\mathrm{tc}}

\newcommand{\co}[1]{[#1)} 
\newcommand{\tl}{[0,\infty)} 

\newcommand{\true}{\mathsf{t}}
\newcommand{\false}{\mathsf{f}}
\newcommand{\unknown}{\bot}

\newcommand{\freezequantifier}{\downarrow}
\newcommand{\freeze}[2]{\mathop{\freezequantifier^{\!#1}}{\!#2}\mathbin{\!.}}
\newcommand{\freezeshort}[1]{\mathop{\freezequantifier}{\!#1}\mathbin{\!.}}

\newcommand{\historically}{\operatorname{\blacksquare}}
\newcommand{\once}{\operatorname{\text{\raisebox{-.11em}{\begin{turn}{45}\scalebox{.9}{$\blacksquare$}\end{turn}}}}}

\newcommand{\until}{\mathbin{\mathsf{U}}}
\newcommand{\always}{\operatorname{\square}}
\newcommand{\eventually}{\operatorname{\text{\raisebox{-.11em}{\begin{turn}{45}\scalebox{.9}{$\square$}\end{turn}}}}}

\newcommand{\weakuntil}{\mathbin{\mathsf{W}}}

\newcommand{\Two}{\mathsf{2}}
\newcommand{\Three}{\mathsf{3}}

\newcommand{\omodels}{\mathrel|\joinrel\approx}

\newcommand{\sem}[2]{\llbracket{#1}\models{#2}\rrbracket}   
\newcommand{\esem}[2]{[{#1}\models{#2}]}                    
\newcommand{\osem}[2]{\llbracket{#1}\omodels{#2}\rrbracket} 
\newcommand{\eosem}[2]{[{#1}\omodels{#2}]}                  

\newcommand{\emap}{[\,]} 

\lstdefinelanguage{myML}{
  backgroundcolor = \color{lightgray!20!},
  firstnumber = 1,
  morekeywords={let,in,if,then,else,otherwise,and,or,not,return,switch,case,fun,raise,loop,for,each,foreach,do,with,while,proc,match,until,default,procedure},
  columns=fullflexible,
  sensitive=true,
  commentstyle = \itshape\color{blue},
  morecomment={[l]\#},
  mathescape=true,
  basicstyle=\scriptsize,
  identifierstyle={\sffamily},
  escapechar = \&
}

\newcommand{\ls}{\lstinline[basicstyle=\normalsize]}
 
\newcommand{\lsfootnotesize}{\lstinline[basicstyle=\footnotesize]}
\newcommand{\lsf}{\lsfootnotesize}

\DeclareFontFamily{U}{mathb}{\hyphenchar\font45}
\DeclareFontShape{U}{mathb}{m}{n}{
      <5> <6> <7> <8> <9> <10> gen * mathb
      <10.95> mathb10 <12> <14.4> <17.28> <20.74> <24.88> mathb12
      }{}
\DeclareSymbolFont{mathb}{U}{mathb}{m}{n}

\DeclareMathSymbol{\sqsubset}{3}{mathb}{"80}
\DeclareMathSymbol{\sqsubseteq}{3}{mathb}{"84}
\DeclareMathSymbol{\sqsubsetneq}{3}{mathb}{"88}

\newcommand{\tpp}{\mathit{tp}}
\newcommand{\tcp}{\mathit{tc}}
\newcommand{\pos}{\position}
\newcommand{\AP}{\mathit{AP}}
\newcommand{\sub}{\mathit{Sub}}

\newcommand{\fPhi}[2]{\Phi_{#1}^{#2}}
\newcommand{\fPsi}[2]{\Psi_{#1}^{#2}}
\newcommand{\fPsip}[2]{\bar\Psi_{#1}^{#2}}

\newcommand{\deltap}{\hat\delta}

\newcommand{\overbar}[1]{\mkern 1.5mu\overline{\mkern-4.5mu #1\mkern-1.5mu}\mkern 1.5mu}
\newcommand{\bAP}{\overbar{\AP}}

\newcommand{\byref}[2]{\text{\small [by (#1): $#2$]}}
\newcommand{\byvar}[2]{\text{\small [by #1: $#2$]}}

\newcommand{\Val}{\mathit{Val}}

\newcommand{\blue}[1]{{\color{blue}#1}}


%% file: intro.tex
\section{Introduction}
\label{sec:intro}

Verifying systems at runtime can be accomplished by instrumenting
system components so that they inform monitors about the actions they
perform. The monitors update their states according to the information
received and check whether the properties they are monitoring are
fulfilled or violated. Various runtime-verification approaches exist
for different kind of systems and property specification languages,
see for example
\cite{Barringer_etal:eagle,Bauer_etal:rv_tltl,Meredith_etal:mop,Basin_etal:rv_mfotl,Sen_etal:decentralized_disitributed_monitoring,MN04-signals,Bauer_Falcone:decentralised_monitor}.

Many of these specifications languages are based on temporal logics or
finite-state machines, which describe the correct system behavior in
terms of \emph{infinite} streams of system actions.  However, at any
point in time, a monitor has only partial knowledge about the system's
behavior.  In particular, a monitor can at best only be aware of the
previously performed actions, which correspond to a finite prefix of
the infinite action stream.  When communication channels are
unreliable, a monitor's knowledge about the previously performed
actions may even be incomplete since messages can be lost or delayed
and thus received out of order.
Nevertheless, a monitor should output a verdict promptly when the
monitored property is fulfilled or violated. Moreover, the verdict
should remain correct when some of the monitor's knowledge gaps are
subsequently closed.

Many runtime-verification approaches rely on an extension of the
standard Boolean semantics of the linear-time temporal logic LTL with
a third truth value, proposed by Bauer et
al.~\cite{Bauer_etal:ltl_rv}.  Namely, a formula evaluates to the
Boolean truth value~$b$ on a finite stream of performed
actions~$\sigma$ if the formula evaluates to $b$ on all infinite
streams that extend~$\sigma$; otherwise, the formula's truth value is
unknown on $\sigma$.  This semantics, however, only accounts for
settings where monitors are always aware of all previously performed
actions.  It is insufficient to reason soundly and completely about
system behavior at runtime when, for example, unreliable channels are used to
inform the monitors about the performed actions.

In this paper, we present an extension of the propositional real-time
logic
MTL~\cite{Koymans:realtime_properties,AlurHenzinger:realtimelogics_survey},
which we name \MTLdata.
First, \MTLdata comprises a freeze
quantifier~\cite{Henzinger:half_order} for reasoning about data values
in action streams.  The freeze quantifier~$\downarrow$ can be seen as
a restricted version of the first-order quantifiers~$\exists$
and~$\forall$.  More concretely, at a position of the action stream,
the formula~$\freeze{}{x}\phi$ uniquely binds a data value of the
action at that position to the logical variable~$x$.

Second, we equip \MTLdata with a new three-value semantics that is
well suited for settings where system components communicate with the
monitors over unreliable channels.  Specifically, we define the
semantics of \MTLdata's connectives over the three truth
values~$\true$,~$\false$, and~$\unknown$.  We interpret these truth
values as in Kleene logic and conservatively extend the logic's
standard Boolean semantics, where $\true$ and $\false$ stand for
``true'' and ``false'' respectively, and the third truth
value~$\unknown$ stands for ``unknown'' and accounts for the monitor's
knowledge gaps.  The models of \MTLdata are finite words where
knowledge gaps are explicitly represented.  Intuitively, a finite word
corresponds to a monitor's knowledge about the system behavior at a
given time and the knowledge gaps may result from message delays,
losses, crashed components, and the like.  Critically in our setting,
reasoning is monotonic with respect to the partial order on truth
values, where $\unknown$ is less than $\true$ and $\false$, and
$\true$ and $\false$ are incomparable.  This monotonicity property
guarantees that closing knowledge gaps does not invalidate previously
obtained Boolean truth values.

Third, we present an online algorithm for verifying systems at runtime
with respect to \MTLdata specifications.  Our algorithm is based on,
and extends, the algorithm for MTL by Basin et
al.~\cite{Basin_etal:failureaware_rv} to additionally handle the
freeze quantifier.  The algorithm's output is sound and complete for
\MTLdata's three-valued semantics and with respect to the monitor's
partial knowledge about the performed actions at each point in time.

Our algorithm works roughly as follows.
It receives messages from the system components describing
the actions they perform.  As with the algorithm
in~\cite{Basin_etal:failureaware_rv}, no assumptions are made on the
order in which messages are received.
The algorithm updates its state for each received
message.  This state comprises a graph structure for reasoning about the
system behavior, i.e., computing verdicts about the monitored
property's fulfillment.  The graph's nodes store the truth values of
the subformulas at the different times for the data values to which
quantified variables are frozen. In each update, the algorithm
propagates data values down to the graph's leaves and propagates
Boolean truth values for subformulas up along the graph's edges.  When
a Boolean truth value is propagated to a root node of the graph,
the algorithm outputs a verdict.

Our main contribution is a runtime-verification approach that makes no
assumptions about message delivery.
It handles a significantly richer specification language than previous
approaches, namely, an extension of the real-time logic MTL with a
quantifier for reasoning about the data processed by the monitored
system.
Furthermore, our approach guarantees sound and complete reasoning with
partial knowledge about system behavior.
Finally, we experimentally evaluate the performance of a prototype
implementation of our approach, illuminating its current capabilities,
tradeoffs, and performance limitations.

The remainder of this paper is structured as follows.
In Section~\ref{sec:prelim}, we introduce relevant notation and
terminology.  In Section~\ref{sec:mtl}, we extend MTL with the freeze
quantifier and give the logic's semantics.  In
Section~\ref{sec:monitor}, we describe our monitoring approach,
including its algorithmic details.  In Section~\ref{sec:casestudy}, we
report on our experimental evaluation.  Finally, in
Sections~\ref{sec:related} and~\ref{sec:concl}, we discuss related
work and draw conclusions. 
Further details are given in the appendixes.


%% file: prelim.tex
\section{Preliminaries}
\label{sec:prelim}

In this section, we introduce relevant notation and terminology.

\paragraph{Intervals.}

An \emph{interval} $I$ is a nonempty subset of $\Qpos$ such that if
$a,b\in I$ then $c\in I$, for any $c\in\Qpos$ with $a\leq c\leq b$. We
use standard notation and terminology for intervals. For example,
$(a,b]$ denotes the interval that is left-open with bound~$a$ and
right-closed with bound~$b$.  Note that an interval $I$ with
cardinality $|I|=1$ is a singleton $\set{\tau}=[\tau,\tau]$, for some
$\tau\in\Qpos$.  An interval~$I$ is \emph{unbounded} if its right
bound is~$\infty$, and \emph{bounded} otherwise.  Let
$I-J := \setx{\tau-\tau'}{\tau\in I\text{ and }\tau'\in J}\cap\Qpos$.

\paragraph{Partial Functions.}

For a partial function $f:A\pto B$, let $\pdef(f):=\setx{a\in
  A}{f(a)\text{ is defined}}$.  If $\pdef(f)=\{a_1,\dots,a_n\}$, for
some $n\in\Nat$, we also write $[a_1\mapsto f(a_1),\dots, a_n\mapsto
f(a_n)]$ for $f$, when $f$'s domain $A$ and its codomain $B$ are
irrelevant or clear from the context.  Note that $[\,]$ denotes the
partial function that is undefined everywhere.  Furthermore, for
partial functions $f, g:A\pto B$, we write $f\sqsubseteq g$ if
$\pdef(f)\subseteq \pdef(g)$ and $f(a)=g(a)$, for all $a\in\pdef(f)$.
We write $f[a\mapsto b]$ to denote the update of a partial function
$f:A\pto B$ at $a\in A$, i.e., $f[a\mapsto b]$ equals $f$, except that
$a$ is mapped to $b$ if $b\in B$, and $a\not\in\pdef(f[a\mapsto b])$
if $b\notin B$.

\paragraph{Truth Values.}

Let $\Three$ be the set~$\set{\true,\false,\bot}$, where
$\true$~(true) and $\false$~(false) denote the standard Boolean
values, and $\unknown$ denotes the truth value ``unknown.''
Table~\ref{tab:connectives} shows the truth tables of some standard
logical operators over $\Three$.  Observe that these operators
coincide with their Boolean counterparts when restricted to the
set~$\Two:=\{\true,\false\}$.
\begin{table}[t]
  \caption{Truth tables for three-valued logical operators 
    (strong Kleene logic).}
  \label{tab:connectives}
  \centering
  \scalebox{1}{
    \renewcommand{\arraystretch}{.8}
  \begin{tabular}{c@{\ } | @{\ }c }
    $\neg$   &  \\
    \hline
    $\true$  & $\false$ \\
    $\false$ & $\true$ \\
    $\unknown$   & $\unknown$ \\
  \end{tabular}
  \qquad
  \begin{tabular}{ c@{\ } | @{\ }c  c  c }
    $\vee$   & $\true$ & $\false$ & $\unknown$ \\
    \hline
    $\true$  & $\true$ & $\true$  & $\true$\\
    $\false$ & $\true$ & $\false$ & $\unknown$\\
    $\unknown$   & $\true$ & $\unknown$   & $\unknown$\\
  \end{tabular}
  \qquad
  \begin{tabular}{ c@{\ } | @{\ }c  c  c}
    $\wedge$ & $\true$ & $\false$ & $\unknown$ \\
    \hline
    $\true$  & $\true$  & $\false$ & $\unknown$\\
    $\false$ & $\false$ & $\false$ & $\false$\\
    $\unknown$   & $\unknown$   & $\false$ & $\unknown$
  \end{tabular}
  \qquad
  \begin{tabular}{ c@{\ } | @{\ }c  c  c}
    $\rightarrow$ & $\true$ & $\false$ & $\unknown$ \\
    \hline
    $\true$  & $\true$ & $\false$ & $\unknown$  \\
    $\false$ & $\true$ & $\true$  & $\true$ \\
    $\unknown$   & $\true$ & $\unknown$   & $\unknown$  
  \end{tabular}
  }
\end{table}
We partially order the elements in $\Three$ by their knowledge:
$\unknown\prec\true$ and $\unknown\prec\false$, and $\true$ and $\false$
are incomparable as they carry the same amount of knowledge.  Note that
$(\Three,\prec)$ is a lower semilattice where $\curlywedge$ denotes the
meet.  
We remark that the operators in Table~\ref{tab:connectives} are
monotonic.  This ensures that reasoning is monotonic in
knowledge. Intuitively, when closing a knowledge gap, represented
by~$\unknown$, with $\true$ or $\false$, we never obtain a truth value
that disagrees with the previous one.

\paragraph{Timed Words.}

Let $\Sigma$ be an alphabet.  A \emph{timed word} over $\Sigma$ is an
infinite word
$(\tau_0,a_0)(\tau_1,a_1)\ldots\in(\Qpos\times\Sigma)^\omega$, where
the sequence of $\tau_i$s is strictly monotonic and nonzeno, that
is, $\tau_i<\tau_{i+1}$, for every $i\in\Nat$, and for every
$t\in\Qpos$, there is some $i\in\N$ such that $\tau_i>t$.


%% file: mtl.tex
\section{Metric Temporal Logic Extensions}
\label{sec:mtl}

In this section, we extend the propositional real-time logic
MTL~\cite{Koymans:realtime_properties,AlurHenzinger:realtimelogics_survey}
with a freeze quantifier~\cite{Henzinger:half_order}.  The logic's
three-valued semantics conservatively extends the standard
Boolean semantics and accounts for 
knowledge gaps during monitoring.

\subsection{Syntax}

Let $P$ be a finite set of predicate symbols, where $\iota(p)$ denotes
the arity of $p\in P$. Furthermore, let $V$ be a set of variables and
$R$ a finite set of registers.
The syntax of the real-time logic \MTLdata is given by the grammar:
\begin{equation*}
  \phi
  \mathbin{\ ::=\ }
  \true \mathbin{\,\big|\,} 
  p(x_1,\dots,x_{\iota(p)}) \mathbin{\,\big|\,}
  \freeze{r}{x}\phi \mathbin{\,\big|\,} 
  \neg\phi \mathbin{\,\big|\,}
  \phi\vee\phi \mathbin{\,\big|\,}
  \phi\until_I\phi 
  \,,
\end{equation*}
where $p\in P$, $x, x_1,x_2\dots,x_{\iota(p)}\in V$, $r\in R$, and $I$
is an interval.  For the sake of brevity, we limit ourselves to the
future fragment and omit the temporal connective for ``next.''
A formula is \emph{closed} if each variable occurrence is bound by a
freeze quantifier.
A formula is \emph{temporal} if the connective at the root of the
formula's syntax tree is $\until_I$.
We denote by~$\sub(\phi)$ the set of $\phi$'s subformulas.

We employ standard syntactic sugar. For example, $\phi\rightarrow\psi$
abbreviates $(\neg\phi)\vee\psi$, and $\eventually_I\phi$
(``eventually'') and $\always_I\phi$ (``always'') abbreviate
$\true\until_I\phi$ and $\neg\eventually_I\neg\phi$, respectively.
The nonmetric variants of the temporal connectives are also easily
defined, e.g., $\always\phi:=\always_{[0,\infty)}\phi$.  Finally, we
use standard conventions concerning the connectives' binding strength
to omit parentheses.  For example, $\neg$ binds stronger than
$\wedge$, which binds stronger than $\vee$, and the connectives
$\neg$, $\vee$, etc. bind stronger than the temporal connectives,
which bind stronger than the freeze quantifier.  To simplify notation,
we omit the superscript~$r$ in formulas like $\freeze{r}{x}\phi$
whenever $r\in R$ is irrelevant or clear from the context.

\begin{example}
  \label{ex:mtl}
  Before defining the logic's semantics, we provide some intuition.
  The following formula formalizes the policy that whenever a customer
  executes a transaction that exceeds some threshold (e.g.~\$2,000) then this customer
  must not execute any other transaction for a certain period of time
  (e.g.~3 days).
  \begin{equation*}
    \always 
    \freeze{\mathit{cid}}{c}\freeze{\mathit{tid}}{t}\freeze{\mathit{sum}}{a}
    \mathit{trans}(c,t,a)\wedge a\geq2000 \to
    \always_{(0,3]}\freeze{\mathit{tid}}{t'}\freeze{\mathit{sum}}{a'}
    \neg\mathit{trans}(c,t',a')
  \end{equation*}
  We assume that the predicate symbol $\mathit{trans}$ is interpreted
  as a singleton relation or the empty set at any point in time.  For instance,
  the interpretation $\set{(\mathit{Alice},42,99)}$ of
  $\mathit{trans}$ at time~$\tau$ describes the action of $\mathit{Alice}$
  executing a
  transaction with identifier~$42$ with the
  amount \$99 at time~$\tau$.  When the interpretation is the empty
  set, no transaction is executed.
  We further assume that when the interpretation of the predicate
  symbol $\mathit{trans}$ is nonempty, the registers $\mathit{cid}$,
  $\mathit{tid}$, and $\mathit{sum}$ store (a)~the transaction's
  customer, (b)~the transaction identifier, and (c)~the transferred amount, respectively.  If the interpretation is the empty
  set, the registers store a dummy value, representing undefinedness.

  The variables $c$, $t$, $a$, $t'$, and~$a'$ are frozen to the
  respective register values.
  For example, $c$ is frozen to the value stored in the
  register~$\mathit{cid}$ at each point in time and is used to
  identify later transactions from this customer.  Furthermore, note
  that, e.g., the variables $t$ and~$t'$ are frozen to values stored
  in the registers~$\mathit{tid}$ at different times.
  The freeze quantifier can be seen as a weak form of the standard
  first-order quantifiers~\cite{Henzinger:half_order}. Since a
  register stores exactly one value at any time, it is irrelevant
  whether we quantify existentially or universally over a register's
  value.
  \qed
\end{example}

\subsection{Semantics}

\MTLdata's models under the three-valued semantics are finite words
(see Definition~\ref{def:observation} below). Such a model represents
a monitor's partial knowledge about the system behavior at a given
point in time.  This is in contrast to the models for the standard
Boolean semantics for MTL, which are infinite timed words and capture
the complete system behavior in the limit.
 
\begin{definition}
  \label{def:observation}
  Let $D$ be the \emph{data domain}, a nonempty set of values with
  $\unknown\not\in D$.
  \emph{Observations} are finite words with letters of the form
  $(I,\sigma,\rho)$, where $I$ is an interval, $\sigma:P\pto
  2^{\bigcup_{\iota\in\Nat}D^\iota}$, and $\rho:R\pto D$.  We define
  observations inductively.
  \begin{itemize}[--]
  \item The word $\big([0,\infty), [\,], [\,]\big)$ of length~$1$
    is an observation.
  \item If $w$ is an observation, then the word obtained by
    applying one of the following transformations to $w$ is an
    observation.
    \begin{enumerate}[(T1)]
    \item \label{enum:observation_split} Some letter $(I,\sigma,\rho)$
      of $w$, where $|I|>1$, is replaced by the three-letter word
      ${\big(I\cap [0,\tau),\sigma,\rho\big)}
      {\big(\set{\tau},\sigma,\rho\big)} {\big(I\cap
        (\tau,\infty),\sigma,\rho\big)}$, where $\tau\in I$ and
      $\tau>0$.  If $\tau=0$, then $(I,\sigma,\rho)$ is replaced by
      $\big(\set{\tau},\sigma,\rho\big)\big(I\cap(\tau,\infty),\sigma,\rho\big)$.
    \item \label{enum:observation_removal} Some letter
      $(I,\sigma,\rho)$ of $w$, where $|I|>1$ and $I$ is bounded, is removed.
    \item \label{enum:observation_data} Some letter $(I,\sigma,\rho)$ of $w$, 
      where $|I|=1$, is replaced by $(I,\sigma',\rho')$, where
      $\sigma\sqsubseteq\sigma'$ and $\rho\sqsubseteq\rho'$, and
      $\sigma\neq\sigma'$ or $\rho\neq\rho'$.
    \end{enumerate}
  \end{itemize}
\end{definition}
For an observation $w$ of length $n\in\Nat$, let
$\position(w):=\{0,\dots,n-1\}$.  We call $i\in \position(w)$ a
\emph{time point} in $w$ if the interval $I_i$ of the letter at
position $i$ in $w$ is a singleton.  In this case, the element of
$I_i$ is the \emph{timestamp} of the time point~$i$, denoted by
$\timestamp_w(i)$.
We note that for any letter $(I,\sigma,\rho)$ of an observation, if
$|I|>1$ then $\sigma=\rho=[\,]$.

\begin{example}
  \label{ex:trans-observation}
  A monitor's initial knowledge is represented by the observation $w_0
  = \big([0,\infty),[\,],[\,]\big)$.
  Suppose a transaction of $\$99$ with identifier~$42$  from $\mathit{Alice}$
  is executed at time~$3.0$.
  The monitor's initial knowledge $w_0$ is then updated by~(T1)
  and~(T3) to $w_1 = \big([0,3.0),[\,],[\,]\big)
  \big(\{3.0\},\sigma,\rho\big) \big((3.0,\infty),[\,],[\,]\big)$,
  where $\sigma(\mathit{trans}) = \set{(\mathit{Alice}, 42, 99)}$ and
  $\rho = [\mathit{cid}\mapsto\mathit{Alice}, \mathit{tid}\mapsto 42,
  \mathit{sum}\mapsto 99]$.
  If the monitor also receives the information that no action has
  occurred in the interval $[0,3.0)$, then its updated knowledge is
  represented by $\big(\{3.0\},\sigma,\rho\big)
  \big((3.0,\infty),[\,],[\,]\big)$, obtained from $w_1$ by~(T2).
  The information that no action has occurred in an interval can be
  communicated explicitly or implicitly by the monitored system to the
  monitor, for instance, by attaching a sequence number to each
  action.  See~\cite{Basin_etal:failureaware_rv} for details.
  Finally, note that the interval of the last letter of any
  observation is always unbounded. This reflects that a
  monitor is unaware of what it will observe in the future.
  \qed
\end{example}

\begin{definition}
  \label{def:refinement}
  The observation~$w'$ \emph{refines} the observation~$w$, written
  $w\sqsubset_1 w'$, iff $w'$ is obtained from $w$ by one of
  the transformations (T\ref{enum:observation_split}),
  (T\ref{enum:observation_removal}), or
  (T\ref{enum:observation_data}).  The reflexive-transitive closure of
  $\sqsubset_1$ is $\sqsubseteq$.
\end{definition}

\MTLdata's three-valued semantics is defined by a function
$\phi\mapsto \osem{w,i,\nu}{\phi}\in\Three$, for a given
observation~$w$, time point~$i\in\N$, and partial
valuation~$\nu:V\pto D$.
We define this function inductively over the formula structure.
For a predicate symbol~$p\in P$, we write in the following
$p(\bar{x})$ instead of $p(x_1,\dots,x_{\iota(p)})$. Furthermore, we
abuse notation by abbreviating, e.g., $\nu(x_1),\dots,\nu(x_n)$ as
$\nu(\bar{x})$, for a partial valuation~$\nu:V\pto D$ and
variables~$x_1,\dots,x_n$.  Also, the notation $\bar{x}\in\pdef(\nu)$
means that $x\in\pdef(\nu)$, for each $x$ occurring in $\bar{x}$.
Finally, we identify the logic's constant symbol~$\true$ with the
Boolean value $\true\in\Three$, and the connectives~$\neg$ and $\vee$
with the corresponding three-valued logical operators in
Table~\ref{tab:connectives}.
\begin{align*}
  \osem{w,i,\nu}{\true}
  :=\ &
  \true
  \\
  \osem{w,i,\nu}{p(\bar{x})}
  :=\ &
  \begin{cases}
    \true & \text{if $\bar{x}\in\pdef(\nu)$, $p\in\pdef(\sigma_i)$,
      and $\nu(\bar{x})\in\sigma_i(p)$}
    \\      
    \false & \text{if $\bar{x}\in\pdef(\nu)$, $p\in\pdef(\sigma_i)$,
      and $\nu(\bar{x})\not\in\sigma_i(p)$}
    \\
    \unknown & \text{otherwise}
  \end{cases}
  \\
  \osem{w,i,\nu}{\freeze{r}{x}\phi}
  :=\ &
  \osem{w,i,\nu[x\mapsto \rho_i(r)]}{\phi}
  \\
  \osem{w,i,\nu}{\neg\phi}
  :=\ &
  \neg\osem{w,i,\nu}{\phi}
  \\
  \osem{w,i,\nu}{\phi\vee\psi}
  :=\ &
  \osem{w,i,\nu}{\phi}
  \vee
  \osem{w,i,\nu}{\psi}
  \\
  \osem{w,i,\nu}{\phi\until_I\psi}
  :=\ &
  \textstyle\bigvee_{j\in\position(w), j\geq i}
  \Big(\istp_w(j) \wedge \tc_{w,I}(j,i) \wedge 
  \osem{w,j,\nu}{\psi} 
  \,\wedge
  \\[-.2cm] & \qquad\qquad\qquad\qquad\qquad
  \textstyle\bigwedge_{i\leq k<j} \big(\istp_w(k)\rightarrow
  \osem{w,k,\nu}{\phi}\big)\Big)
\end{align*}
The auxiliary functions $\istp_w:\position(w)\rightarrow\Three$ and
$\tc_{w,I}:\position(w)\times\position(w)\rightarrow\Three$, are
defined as follows, where $I_k$ denotes the interval at position
$k\in\position(w)$ in $w$.
\begin{align*}
  \istp_w(j) & :=\begin{cases}
    \true & \text{if $j$ is a time point in $w$}
    \\
    \unknown & \text{otherwise}
  \end{cases}
  \\
  \tc_{w,I}(i,j) & :=\begin{cases}
    \true & 
    \text{if $\tau-\tau'\in I$, for all $\tau \in I_i$ and $\tau'\in I_j$}
    \\
    \false & 
    \text{if $\tau-\tau'\notin I$, for all $\tau \in I_i$ and $\tau'\in I_j$}
    \\
    \unknown & 
    \text{otherwise}
  \end{cases}
\end{align*}

We comment on the semantics of $\phi\until_I\psi$.
The auxiliary functions account for the positions in $w$ that are not time
points. For example, at position $i$, for a position $j\leq i$ to be a
``valid anchor'' for the formula, $j$ must be a time point (in this
case $\istp_w(j)=\true$). Otherwise, the truth value $\unknown$ is
used to express that it is not yet known whether the
interval at position $j$ in $w$ will contain a time point. Note that
using the truth value $\false$ would be incorrect since a refinement
of $w$ might contain a time point with a timestamp in $I_j$.
Furthermore, $\tc_{w,I}(i,j)$ is used to account for the metric
constraint of the temporal connective.  In particular,
$\tc_{w,I}(i,j)$ is $\unknown$ if it is unknown in $w$ whether
the formula's metric constraint is always satisfied or never satisfied
for the positions $i$ and $j$.
Finally, suppose that $\phi$'s truth value is $\false$ at a position
$k$ between $j$ and $i$.  If the interval $I_k$ at position~$k$ is
not a singleton, the function $\istp_w(k)$ ``downgrades'' this value
to $\unknown$, since it will be irrelevant in refinements of $w$ that
do not contain any time points with timestamps in $I_k$.

Note that it may be the case that $\osem{w,i,\nu}{\phi}\in\Two$ when
$i$ is not a time point in $w$ (i.e., $I_i$ is not a singleton).  A
trivial example is when $\phi=\true$.  In a refinement of $w$, it
might turn out that there are no time points with timestamps in $I_i$,
and hence a monitor should not output a verdict for the
specification~$\phi$ at position~$i$ in $w$.  We address this artifact
by downgrading (with respect to the partial order~$\prec$) a Boolean
truth value $\osem{w,i,\nu}{\phi}$ to $\unknown$ when $i$ is not a
time point. To this end, we introduce the following variant of the
semantics.
\begin{definition}
For a formula $\phi$, an observation $w$, $\tau\in\Qpos$, and $\nu$ a
partial valuation, we define $\eosem{w, \tau,\nu}{\phi}:=\osem{w,
  i,\nu}{\phi}$, provided that $\tau$ is the timestamp of some time
point $i\in\position(w)$ in $w$, and $\eosem{w,
  \tau,\nu}{\phi}:=\unknown$, otherwise.
\end{definition}

\subsection{Properties}

The following theorem states that \MTLdata's three-valued semantics is
monotonic in $\sqsubseteq$ (on observations and partial valuations) and
$\preceq$ (on truth values). This property is crucial for monitoring
since it guarantees that a verdict output for an observation stays
valid for refined observations.
\begin{theorem}
  \label{thm:monotonicity}
  Let $\phi$ be a formula, $\mu$ and $\nu$ partial valuations, $u$
  and $v$ observations, and $\tau\in\Qpos$. If $u\sqsubseteq v$ and
  $\mu\sqsubseteq \nu$ then
  $\eosem{u,\tau,\mu}{\phi} \preceq \eosem{v,\tau,\nu}{\phi}$.
\end{theorem}

A similar theorem shows that \MTLdata's three-valued semantics
conservatively extends the standard Boolean semantics (see
Appendix~\ref{app:mtl_bool} for details). Intuitively speaking, if a
formula~$\phi$ evaluates to a Boolean value for an observation at time
$\tau\in\Qpos$, then $\phi$ has the same Boolean value at time~$\tau$
for any timed word\footnote{We assume here that the timed words are
  over the alphabet $\Sigma$ that consists of the pairs
  $(\sigma,\rho)$, where (i)~$\sigma$ is a total function over $P$
  with $\sigma(p)\subseteq D^{\iota(p)}$ for $p\in P$, and (ii)~$\rho$
  is a total function over $R$ with $\rho(r)\in D$ for $r\in R$.  }
that refines the observation.  Formally, a timed word $w'$
\emph{refines} an observation~$w$, $w\sqsubseteq w'$ for short, if for
every $j\in\Nat$, there is some $i\in\position(w)$, such that
$\tau_j\in I_i$, $\sigma_i\sqsubseteq \sigma'_j$, and
$\rho_i\sqsubseteq\rho'_j$, where $(I_\ell,\sigma_\ell,\rho_\ell)$ and
$(\tau_k,\sigma'_k,\rho'_k)$, for $\ell\in\position(w)$ and
$k\in\Nat$, are the letters of $w$ and $w'$, respectively.

We investigate next the decision problem that underlies monitoring.
\begin{theorem}
  \label{thm:decidable}
  For an arbitrary formula~$\phi$, observation~$w$, partial
  valuation~$\nu$, time $\tau\in\Qpos$, and truth value $b\in\Two$,
  the question of whether $\eosem{w,\tau,\nu}{\phi}$ equals~$b$ is
  $\mathrm{PSPACE}$-complete.
\end{theorem}
In a propositional setting, the corresponding decision problem can be
solved in polynomial time using dynamic programming,
where the truth values at the positions of an observation are
propagated up the formula structure.  Note that the truth value of a
proposition at a position is given by the observation's letter at that
position.
This is in contrast to \MTLdata, where atomic formulas can have free
variables and their truth values at the positions in an
observation~$w$ may depend on the data values stored in the registers
and frozen to these variables at different time points of $w$.  Before
truth values are propagated up,  the bindings of variables to data
values must be propagated down.


%% file: new_monitor.tex
\section{Monitoring Algorithm}
\label{sec:monitor}

In this section, we present an online algorithm that computes verdicts
for \MTLdata specifications.  To support scalable monitoring, the
computation is incremental in that, when refining an observation
according to the
transformations~(T\ref{enum:observation_split})--(T\ref{enum:observation_data}),
the results from previous computations are reused, including the
propagated data values and Boolean values.  We also define correctness
requirements for monitoring and establish the algorithm's correctness.

\subsection{Correctness Requirements}
\label{subsec:requirements}

We define when a sequence of observations is valid for representing a
monitor's knowledge over time.  We assume that the monitor receives
in the limit infinitely many messages containing information about the system
behavior.  This assumption is invalid if the system ever terminates.
Nevertheless, we make this assumption to simplify matters and
it is easy to adapt the definitions and results to the general case.
\begin{definition}
  The infinite sequence $\bar{w} = (w_i)_{i\in\Nat}$ of observations
  is \emph{valid} if $w_0 = ([0,\infty), [\,], [\,])$ and
  $w_i\sqsubsetneq w_{i+1}$, for all $i\in\Nat$.
\end{definition}

Let $M$ be a monitor and $\bar{w}$ a valid sequence of observations.
In the following, we view $w_i$ as the input to $M$ at iteration $i$.
For the input $w_i$, $M$ outputs a set of \emph{verdicts}, which is a
finite set of pairs $(\tau,b)$ with $\tau\in\Qpos$ and $b\in\Two$. We
denote this set by $M(w_i)$.
Note that in practice, $M$ would receive at iteration $i>0$ a message
that describes just the differences between $w_{i-1}$ and $w_i$.
Furthermore, the $w_i$s can be understood as abstract descriptions of
$M$'s states over time, representing $M$'s knowledge about the system
behavior, where $w_0$ represents $M$'s initial knowledge.  Also note
that if the timed word $v$ is the system behavior in the limit, then
$w_i\sqsubseteq v$, for all $i\in\Nat$, assuming that components do
not send bogus messages. However, for every $i\in\Nat$, there are
infinitely many timed words~$u$ with $w_i\sqsubseteq u$.  Since
messages sent to the monitor can be lost, it can even be the case that
there are timed words $u$ with $u\not=v$ and $w_i\sqsubseteq u$, for
all $i\in\Nat$.

\begin{definition}
  \label{def:soundness_completeness-observation}
  Let $M$ be a monitor, $\phi$ a formula, and $\bar{w}$ a valid
  observation sequence.
  \begin{itemize}[--]
  \item $M$ is \emph{observationally sound} for $\bar{w}$ and $\phi$
    if for all partial valuations~$\nu$ and $i\in\N$, if
    $(\tau,b)\in M(w_i)$ then $\eosem{w_i,\tau,\nu}{\phi}=b$.
  \item $M$ is \emph{observationally complete} for $\bar{w}$ and
    $\phi$ if for all partial valuations~$\nu$, $i\in\N$, and
    $\tau\in\Q_{\geq0}$, if $\eosem{w_i,\tau,\nu}{\phi}\in\Two$
    then $(\tau,b)\in\bigcup_{j\leq i}M(w_j)$, for some $b\in\Two$.
  \end{itemize}
  We say that $M$ is \emph{observationally sound} if $M$ is
  observational sound for all valid observation sequences and formulas
  $\phi$. The definition of $M$ being \emph{observationally complete}
  is analogous.
\end{definition}

It follows from Theorem~\ref{thm:decidable} that there exist monitors
for \MTLdata that are both observationally sound and complete.
This is in contrast to correctness requirements that demand that a
monitor outputs a verdict as soon as the specification has the same
Boolean value on every extension of the monitor's current knowledge.
It is easy to see that, for a given specification language, such monitoring
is at least as hard as checking satisfiability for the
language.  The propositional fragment of \MTLdata is already
undecidable~\cite{OuaknineW06}.
Thus monitors satisfying such strong requirements do not exist for
\MTLdata.
For LTL, such stronger requirements are standardly formalized using a
three-valued “runtime-verification” semantics, as introduced by Bauer
et al.~\cite{Bauer_etal:rv_tltl}, and adopted by other
runtime-verification approaches, e.g.~\cite{Bauer-FMSD15}.  See
Appendix~\ref{subapp:req_bool} for a formal definition of these
requirements in our setting.

\begin{example}
  Consider the formula $\phi=\always(p \wedge \eventually \neg p)$.
  Under the classical Boolean semantics, $\phi$ is logically
  equivalent to $\false$, however not under our semantics. For
  example, $\osem{w,0,\nu}{\phi}=\bot$, for
  $w=\big([0,\infty),\emap,\emap\big)$ and any valuation~$\nu$.
  Given a valid observation sequence~$\bar{w}$, an observationally
  sound and complete monitor for $\bar{w}$ and $\phi$ will first
  output the verdict $(0,\false)$ for the minimal $i$ such that $w_i$
  contains a letter that assigns $p$ to false.  \qed
\end{example}

\subsection{Monitoring Algorithm}
\label{subsec:alg}

We sketch the algorithm's state, its main procedure, and its main data
structure. We provide further algorithmic details in
Appendix~\ref{app:pseudocode}.

\subsubsection{Monitor State.}
\label{subsubsec:state}

Before explaining the algorithm, we first rephrase the \MTLdata's
semantics such that it is closer to the representation used by the
monitor.
Given an $i\in\Nat$, a position $j\in\position(w_i)$, and a subformula
$\gamma$ of $\phi$, we denote by~$\fPhi{i}{\gamma,J_j}$, where $J_k$
is the interval of the $k$th letter of $w_i$, the
propositional~formula:
\begin{equation*}
  \fPhi{i}{\gamma,J_j} := \left\{
    \begin{array}{l@{\quad}l}
      \gamma^{J_j} & \text{if $\gamma$ is atomic} 
      \\
      \neg\alpha^{J_j} & \text{if $\gamma=\neg\alpha$}
      \\
      \alpha^{J_j}\lor\beta^{J_j} & \text{if $\gamma=\alpha\lor\beta$}
      \\
      \alpha^{J_j} & \text{if $\gamma=\freeze{r}{x}{\alpha}$}
      \\
      \bigvee_{k\geq j}\big(\tpp^{J_k} \land \tcp^{J_k,J_j}_\gamma \land \beta^{J_k} \land
      \bigwedge_{j\leq h<k}(\tpp^{J_h} \to \alpha^{J_h})
      \big)
      & 
      \text{if $\gamma=\alpha\until_I\beta$},
    \end{array}
    \right.
\end{equation*}
where $\alpha^{K}$, $\tpp^{K}$, and $\tcp^{H,K}_\psi$ denote atomic
propositions, for each proper subformula~$\alpha$ of $\phi$, each
temporal subformula $\psi$ of $\phi$, and all intervals $H,K$ of
letters in $w_i$.
Next, we define, for any partial valuation~$\nu$, the
substitution~$\theta_i^\nu$ of Boolean values for these atomic
propositions as follows:
\[\begin{array}{rl@{\quad}l}
  \theta_i^\nu(\alpha^{J_j}) := & \osem{w_i,j,\nu}{\alpha} & \text{if $\osem{w_i,j,\nu}{\alpha}\in\Two$},
  \\
  \theta_i^\nu(\tpp^{J_j}) := & \istp_{w_i}(j) & \text{if $\istp_{w_i}(j)\in\Two$},
  \\
  \theta_i^\nu(\tcp_{\alpha\until_I\beta}^{J_j,J_k}) := & \tc_{w_i,I}(j,k) & \text{if $\tc_{w_i,I}(j,k)\in\Two$},
\end{array}\]
and $\theta_i^\nu$ is undefined otherwise.
In what follows, the symbol $\equiv$ denotes semantic equivalence
between propositional formulas.
It is easy to see that
$$
\theta_i^\mu(\fPhi{i}{\gamma,J_j}) \equiv \osem{w_i,j,\nu}{\gamma}
\quad \text{iff} \quad 
\osem{w_i,j,\nu}{\gamma}\in\Two,
$$
where $\mu = \nu[x\mapsto \rho_j(r)]$ if
$\gamma=\freeze{r}{x}{\alpha}$ and $\mu=\nu$ otherwise, with $\rho_j$
being the third component of the $j$th letter of $w_i$.
Note that the formula $\theta_i^\mu(\fPhi{i}{\gamma,J_j})$ tells us
more than the truth value $\osem{w_i,j,\nu}{\gamma}$. Indeed, when
$\theta_i^\mu(\fPhi{i}{\gamma,J_j})\not\equiv b$, for each $b\in\Two$,
then we know not only that $\osem{w_i,j,\nu}{\gamma}=\bot$, but we
also know what the causes of uncertainty are, namely the direct
subformulas~$\alpha$ of $\gamma$ and indexes~$k$ with
$\osem{w_i,k,\mu}{\alpha}=\bot$.

The monitor maintains as state between its iterations a variant of
the propositional formulas $\theta_i^\mu(\fPhi{i}{\gamma,J_j})$.
The reason for using variants is
that it is not algorithmically convenient to transform
$\theta_i^\mu(\fPhi{i}{\gamma,J})$ into
$\theta_{i+1}^\mu(\fPhi{i+1}{\gamma,K})$, where $K$ is an interval (of a letter) in $w_{i+1}$
that originates from the interval $J$ in $w_i$.
Such a transformation is needed for obtaining an incremental
monitoring algorithm that reuses information already
computed at previous iterations.

The formulas that the monitors maintains, denoted
$\fPsi{i}{\gamma,J_j,\nu}$, can be obtained from the formulas
$\theta_i^\mu(\fPhi{i}{\gamma,J_j})$ as follows.
When $\gamma$ is a nontemporal formula, then
$\fPsi{i}{\gamma,J_j,\nu}$ equals
$\theta_i^\mu(\fPhi{i}{\gamma,J_j})$.
When $\gamma$ is a temporal formula $\alpha\until_I\beta$, then, to
each disjunct for index~$k$ in $\fPhi{i}{\gamma,J_j}$, we add the
subformula $(\tpp^{J_k} \lor \alpha^{J_k})$ as a conjunct. This is
sound, based on the equivalence
$\tpp^{J_k}\equiv \tpp^{J_k} \land (\tpp^{J_k} \lor \alpha^{J_k})$.
Furthermore, the monitor treats the subformulas
$(\tpp^{J_k} \land \beta^{J_k})$, $(\tpp^{J_h} \to \alpha^{J_h})$, and
$(\tpp^{J_k} \lor \alpha^{J_k})$ in a special way: they are not
simplified in $\fPsi{i}{\gamma,J,\nu}$ when they are still needed to
obtain~$\fPsi{i+1}{\gamma,K,\nu}$. That is, even if one the atomic
propositions~$q$ of these subformulas could be instantiated
(i.e. $q\in\pdef(\theta_i^\mu)$) this is not always done, as explained
in the next section.
Instead, these three types of subformulas are represented in
$\fPsi{i}{\gamma,J_j,\nu}$ by the atomic
propositions~$\bar\beta^{J_k}$, $\bar{\alpha}^{J_h}$, and
$\bar{\bar{\alpha}}^{J_k}$, respectively.

\begin{example}\label{ex:alphaK}
  We illustrate here the definitions of the propositional formulas
  $\fPhi{i}{\gamma,J,\nu}$ and $\fPsi{i}{\gamma,J,\nu}$ for temporal
  formulas~$\gamma$. We also suggest why variants of
  the formulas $\theta_i^\mu(\fPhi{i}{\gamma,J_j})$ are needed.

  Let $\gamma = p \until q$, where $p$ and $q$ are $0$-ary
  predicates. Assume that in $w_1$ we have the intervals
  $L=\co{0,\tau_1}$, $N=\set{\tau_1}$, and $R=(\tau_1,\infty)$, and in
  $w_2$ we have the intervals $L_1=\co{0,\tau_0}$, $L_2=\set{\tau_0}$,
  $L_3=(\tau_0,\tau_1)$, $N$, and $R$, with $\tau_0\in L$.
  Assume also that neither $p$ nor $q$ holds at $\tau_1$. Then
  \[
  \begin{array}{rl@{\qquad}rl}
    \theta_1^{\emap}(\fPhi{1}{\gamma,L}) & \equiv \tpp^L \land q^L
    & 
    \theta_2^{\emap}(\fPhi{2}{\gamma,L_2}) & \equiv q^{L_2} \lor (\tpp^{L_3} \land q^{L_3} \land p^{L_2}) 
    \\[0.5ex]
    \fPsi{1}{\gamma,L,\emap} & = \bar{q}^L \land \bar{\bar{p}}^L
    &
    \fPsi{2}{\gamma,{L_2},\emap} & = \bar{q}^{L_2} \lor (\bar{q}^{L_3} \land \bar{\bar{p}}^{L_3} \land \bar{p}^{L_2})
  \end{array}
  \]
  Note that $p^L$ is not an atomic proposition
  of~$\fPhi{1}{\gamma,L}$, while $\bar{\bar{p}}^L$ is an atomic
  proposition of~$\fPsi{1}{\gamma,L,\emap}$.
  This last fact allows the monitoring algorithm to obtain
  $\fPsi{2}{\gamma,{L_2},\emap}$ from $\fPsi{1}{\gamma,{L},\emap}$, by
  introducing the needed new propositions $\bar{p}^{L_2}$, $\bar{p}^{L_3}$, and~$\bar{\bar{p}}^{L_2}$.
  \qed
\end{example}

To recapitulate, the monitor's state at iteration~$i$ consists of
propositional formulas~$\fPsi{i}{\gamma,J,\nu}$, one for each
subformula $\gamma$ of $\phi$, interval $J$ occurring in a letter of
$w_i$, where $i$ is the current iteration, and partial valuation $\nu$
that is \emph{relevant} for the current subformula and position
corresponding to $J$ in~$w_i$.
Intuitively, a valuation~$\nu$ is relevant for $\psi$ and a position
$j\in\pos(w_i)$,  if $\osem{w_i,j,\nu}{\psi}$ is reached when unfolding
the formula that defines $\osem{w_i,k,\emap}{\phi}$, for some
$k\in\pos(w_i)$.\footnote{We consider here that the formulas defining
  the semantics are first simplified. E.g., assuming that
  $\osem{w_i,j,\nu}{\alpha\until_I\beta}$ is reached, $k\in\pos(w_i)$,
  and $k\geq j$, if ${\tc_{w_i,I}(k,j)=\false}$, then
  $\osem{w_i,k,\nu}{\beta}$ is not reached, otherwise
  (i.e.~$\tc_{w_i,I}(k,j)\neq\false$) it is reached.}
For instance, $\emap$ is relevant for $\phi$ and any
$j\in\pos(w_i)$. Furthermore, if $\nu$ is relevant for
$\freeze{r}{x}\psi$ and $j$, then $\nu[x\mapsto\rho_j(r)]$ is relevant
for $\psi$ and $j$.
 
\begin{example}\label{ex:freeze}
    Let $\phi := \freeze{r}{x}\eventually_{(0,1]} p(x)$.  For brevity,
    we treat the temporal connective $\eventually_{(0,1]}$ as a
    primitive.
    Also, for readability, we let $\alpha := \eventually_{(0,1]} p(x)$
    and $\beta := p(x)$.
    Consider an observation $w_1$ that has the same interval structure as in the
    previous example and the second letter is $(\tau_1,\sigma,\rho)$
    with $\rho(r)=d$ for some data value $d$ and $p\notin\pdef(\sigma)$.
    The monitor's state for $w_1$ consists of the formulas:
    \[\begin{array}{rl@{\qquad\ }rlr}
     \fPsi{1}{\phi,K,\emap} & = \alpha^{K},\ \text{for any $K\in\set{L,N,R}$},
     &
     \fPsi{1}{\alpha,L,\emap} & =
     \bar\beta^{L} \lor \bar\beta^{N} \lor \bar\beta^{R},
     &
     \\
     \fPsi{1}{\beta,K,\emap} & = \beta^{K},\ \text{for any $K\in\set{L,N,R}$},
     &
     \fPsi{1}{\alpha,N,[x\mapsto d]} & = \bar\beta^{R},
     &
     \\
     \fPsi{1}{\beta,R,[x\mapsto d]} & = \beta^{R},
     &   
     \fPsi{1}{\alpha,R,\emap} & = \bar\beta^{R}.  
     & 
    \end{array}\]
  Note that there are two relevant valuations for $\beta$ and
  position~2 (which is the position of the interval $R$ in $w_1$),
  namely $\emap$ and $[x\mapsto d]$. This follows from the definition
  and it corresponds to the fact that $\bar\beta^R$ is an atomic proposition
  of a formula both of the form $\fPsi{1}{\alpha,K,\emap}$ (namely,
  when $K\in\set{L,R}$) and of the form
  $\fPsi{1}{\alpha,K,[x\mapsto d]}$ (namely, when $K=N$).
  \qed
\end{example}

\subsubsection{Main Procedure.}
\label{subsubsec:main}

The monitor's pseudocode is shown in Listing~\ref{fig:monitor}.
After initializing the monitor's state, the monitor loops. In each
loop iteration, the monitor receives a message, updates its state
according to the information extracted from the message, and outputs
the computed verdicts.

We assume that each received message describes a new time point in an
observation, i.e., a letter of the form~$(\{\tau\},\sigma,\rho)$.  
Furthermore, we assume that each received message~$m$ contains
information that identifies the \emph{component} that has sent the
message to the monitor and a \emph{sequence number}, i.e., the number
of messages, including $m$, that the component has sent to the monitor
so far.
Using this information, the monitor can detect \emph{complete}
intervals, i.e., the non\-singleton intervals that do not contain the
timestamp of any message that the monitor processes in later
iterations.
Thus, the received messages describe the ``deltas'' of a valid
observation sequence (cf.~Section~\ref{subsec:requirements}), where
the next observation is obtained from the previous one by applying
transformation~(T1), followed by~(T3), possibly followed by several
applications of~(T2).

\begin{listing}[t]
  \centering
  \begin{minipage}[t]{0.15\linewidth}
    \caption{}
    \label{fig:monitor}
  \end{minipage}
  \quad
  \begin{minipage}[t]{0.8\linewidth}
\begin{lstlisting}
procedure Monitor($\phi$)
  Init($\phi$)
  loop
    $m$ $\leftarrow$ NewMessage()
    $\tau$, $\sigma$, $\rho$, comp, seq_num := Parse($m$)
    $J$, new := Split($\tau$, comp, seq_num)
    NewTimePoint($\phi$, $J$, new)
    foreach $\freeze{r}{x}{\psi}$ in $\sub(\phi)$ with $r\in\pdef(\rho)$ do
      PropagateDown($\psi$, $\set{\tau}$, $x$, $\rho(r)$)
    foreach $\fPsi{}{p(\bar{x}),\set{\tau},\nu}$ $\neq$ nil with $\bar{x}\in \pdef(\phi), p\in\pdef(\sigma)$ do
      $b$ := ($\nu(\bar{x})\in \sigma(p)$)
      $\fPsi{}{p(\bar{x}),\set{\tau},\nu}$ := $b$
      PropagateUp($p(\bar{x})$, $\set{\tau}$, $b$)
    NewVerdicts()
\end{lstlisting}
  \end{minipage}
\end{listing}

With the procedure \ls{NewMessage}, the monitor receives a new
message, for instance over a channel or a log file.
Next, the monitor parses the message to recover the corresponding
letter~$(\{\tau\},\sigma,\rho)$, the component, and the sequence
number.
Afterwards, using the procedure~\ls{Split}, the monitor determines
the interval~$J$ that is split (namely, the one
where $\tau\in J$) and the resulting new, incomplete intervals,
stored in the sequence~\ls{new}. Concretely, the intervals in~\ls{new}
consist of those intervals among $J\cap\co{0,\tau}$, $\set{\tau}$, and
$J\cap(\tau,\infty)$ that are not complete. Note that \ls{new}
contains at least the singleton~$\set{\tau}$.
The detection of complete intervals by the \ls{Split} procedure is
done in the same manner as in~\cite{Basin_etal:failureaware_rv}.

The remaining pseudocode updates the monitor's state to reflect
the new observation. It first transforms formulas
$\fPsi{}{\gamma,K,\nu}$ so that they reflect the interval structure of
the new observation, with \ls{NewTimePoint}.
Afterwards, the monitor propagates the new data values down (the
formula $\phi$'s syntax tree) with \ls{PropagateDown}, and propagates newly
obtained Boolean values up with \ls{PropagateUp}.
The procedures \ls{NewTimePoint} and \ls{PropagateUp} are conceptually
similar to analogous procedures given
in~\cite{Basin_etal:failureaware_rv}, although the
formulas~$\fPsi{}{\gamma,K,\nu}$ were implicit
in~\cite{Basin_etal:failureaware_rv}.
We outline next these three procedures and give their pseudocode in
Appendix~\ref{app:pseudocode}.
Finally, the monitor reports the verdicts computed during the current
iteration by calling the procedure \ls{NewVerdicts}. 

In the rest of the section, we use the convention that whenever
$\gamma$ or $\nu$ are not specified in a formula
$\fPsi{}{\gamma,J_j,\nu}$ then we assume they are an arbitrary
subformula of~$\phi$ and respectively an arbitrary partial valuation
that is relevant for $\gamma$ and~$j$.

\paragraph{Adding a New Time Point.}

The procedure \ls{NewTimePoint} builds new formulas
$\fPsi{}{\gamma,K,\nu}$ with $K\in\mathsf{new}$ from the corresponding
formulas $\fPsi{}{\gamma,J,\nu}$. It also updates all formulas
$\fPsi{}{\gamma,J,\nu}$ such that they use atomic propositions
$\alpha^K$ with $K\in\mathsf{new}$ instead of $\alpha^J$.
For nontemporal formulas~$\gamma$, the update is straightforward. For
instance, if $\gamma=\alpha\lor\beta$ and $\fPsi{}{\gamma,J,\nu} =
\beta^J$, then $\fPsi{}{\gamma,K,\nu} = \beta^K$, for each
$K\in\mathsf{new}$.
For temporal formulas~$\gamma$, the update is more involved, although it
can be performed easily by applying well-suited substitutions.
To illustrate the kind of updates that are needed, suppose for
example that $\gamma=\alpha\until_I\beta$ and that
$\fPsi{}{\gamma,J',\nu}$, for some $J'<J$, contains the atomic
proposition $\bar{\alpha}^J$. Then $\fPsi{}{\gamma,J',\nu}$ is updated
by replacing $\bar{\alpha}^J$ with the conjunct
$\bigwedge_{K\in\mathsf{new}}\bar\alpha^K$.
Finally, we note that formulas $\fPsi{}{\gamma,K,\nu}$ with $K\neq J$
and without atomic propositions $\alpha^J$ need not be updated.

\paragraph{Downward Propagation.}

Whenever a variable $x$ is frozen to a data value at time~$\tau$, the
procedure \ls{PropagateDown} updates the monitor's state to account for
this fact. Concretely, this value is propagated according to the
semantics through partial valuations to atomic formulas $p(\bar{y})$.
The propagation is performed by starting from formulas
$\fPsi{}{\freeze{r}{x}{\psi}, \{\tau\}, \mu}$ and recursively visiting
formulas $\fPsi{}{\alpha,K,\nu}$ with $\alpha$ a subformula of
$\psi$. For each visited formula, a new formula
$\fPsi{}{\alpha,K,\nu[x\mapsto \rho(r)]}$ is created, where the new
formula is simply a copy of~$\fPsi{}{\alpha,K,\nu}$.
Note that the old formula $\fPsi{}{\alpha,K,\nu}$ may still be
relevant in the future.
For instance, suppose a value $d$ is propagated from
$\fPsi{}{\eventually_I\beta,\set{\tau},\nu}$ to
$\fPsi{}{\beta,K,\nu}$, copying it to
$\fPsi{}{\beta,K,\nu[x\mapsto d]}$, and suppose also that
$\bar{\beta}^K$ is an atomic proposition in
$\fPsi{}{\eventually_I\beta,J',\nu}$. Then $\fPsi{}{\beta,K,\nu}$
might be used again later when another data value~$d'$ is propagated
downwards from $\fPsi{}{\eventually_I\beta,\set{\tau'},\nu}$ with
$\tau'\in J'$, to copy it to
$\fPsi{}{\beta,K,\nu[x\mapsto d']}$.

\paragraph{Upward Propagation.}

The procedure \ls{PropagateUp} performs the following update of the
monitor's state.
When a formula $\fPsi{}{\alpha,K,\mu}$ simplifies to a Boolean
value~$b$, then this Boolean value is propagated up the syntax tree of
$\phi$ as follows: $\alpha^K$ is instantiated to $b$ in every formula
$\fPsi{}{\gamma,J',\nu}$ that has $\alpha^K$ as an atomic proposition,
except when $\gamma$ is itself an atom of $\phi$. The formula is then
simplified (using rules like $z\lor\true\equiv\true$) and if it
simplifies to a Boolean value then propagation continues recursively.
Note that $\gamma$ is a parent of $\alpha$. When
$\fPsi{}{\phi,\set{\tau'},\emap}$ is simplified to a Boolean value
$b'$, then $(\tau',b')$ is marked as a new verdict.
Propagation starts from the atoms of $\phi$. The Boolean value $\true$
is propagated from the atom~$\true$ only once, in the \ls{Init}
procedure. 
For an atom~$\alpha=p(\bar{x})$, the monitor sets
$\fPsi{}{p(\bar{x}),\set{\tau},\mu}$ to a Boolean value, if possible,
according to the semantics, for all relevant valuations~$\mu$.

Recall that for temporal formulas $\gamma=\alpha\until_I\beta$, the
formula $\fPsi{}{\gamma,J',\nu}$ contains atomic propositions of the
form $\bar{\alpha}^K$, $\bar{\bar\alpha}^K$, and $\bar{\beta}^K$
instead of $\alpha^K$ and $\beta^K$. These atomic propositions are
treated specially: they are not instantiated when $K$ is not a
singleton and the value $b$ to be propagated is $\true$ for $\beta$
formulas and $\false$ for $\alpha$ formulas (otherwise they are
instantiated). This behavior corresponds to the meaning of these
atomic proposition given in Section~\ref{subsubsec:state}.
For instance, $\bar{\beta}^K$ stands for $\tpp^K\lor\beta^K$ and thus
it is not instantiated to $\true$ in $\fPsi{}{\gamma,J',\nu}$ when $K$
is not a singleton even when $\fPsi{}{\gamma,K,\nu}=\true$, because
the existence of a time point in $K$ is not guaranteed: it might turn
out that $K$ is a complete interval.
The propagation will be done later for singletons $\set{\tau'}$ with
$\tau'\in K$, if and when a message with timestamp $\tau'$ arrives.

\subsubsection{Data Structure.}
\label{subsubsec:datastructure}

We have not yet described the data structure used in our pseudocode,
which is needed for an efficient implementation.
The data structure that we use is similar to that described
in~\cite{Basin_etal:failureaware_rv}.
Namely, it is a directed acyclic graph.
The graph's \emph{nodes} are tuples of the form~$(\psi,J,\nu)$, where
$\psi$ is a subformula of~$\phi$, $J$ an interval, and $\nu$ a
partial valuation.
Each node $(\gamma,J,\nu)$ stores the associated propositional formula
$\fPsi{}{\gamma,J,\nu}$.
Nodes are linked via \emph{triggers}: a trigger of a
node~$(\alpha,K,\mu)$ points to a node~$(\gamma,J,\nu)$ if and only if
$\alpha^K$, $\bar{\alpha}^K$, or $\bar{\bar\alpha}^K$ is an atomic
proposition of $\fPsi{}{\gamma,J,\nu}$, $\gamma$ is a nonatomic
formula, and $\mu = \nu[x\mapsto \rho^i_j(r)]$ if
$\gamma=\freeze{r}{x}{\alpha}$ and $\mu=\nu$ otherwise.
Triggers are actually bidirectional: for any (outgoing) trigger
there is a corresponding ingoing trigger.

This data structure allows us to directly access, given a
formula~$\fPsi{}{\alpha,K,\mu}$, all the
formulas~$\fPsi{}{\gamma,J,\nu}$ that have $\alpha^K$ as an atomic
proposition.
Also, conversely, for any formula~$\fPsi{}{\gamma,J,\nu}$ the data
structure allows us to directly access the
formula~$\fPsi{}{\alpha,K,\mu}$ for any atomic proposition $\alpha^K$
of $\fPsi{}{\gamma,J,\nu}$.
These two operations are used for upward and downward propagation
respectively.
We note also that a node for which the associated propositional
formula has simplified to a Boolean value that has been propagated can
be deleted.

\begin{figure}[t]
  \centering
  \scalebox{0.8}{
    \begin{tikzpicture}[thick, x=1pt, y=1pt]
      \tikzstyle{node} = [draw, minimum height = 14pt, minimum width = 14pt]
      \tikzstyle{trigger} = [->]
      \tikzstyle{guard} = [draw, circle, fill, minimum size = 0pt, inner sep = 0pt]
      \newcommand{\x}{0}
      
      \node at (\x,80){$\freezeshort{x}\eventually_{(0,1]}p(x)$};
      \node at (\x,40) {$\phantom{\freezeshort{x}}\eventually_{(0,1]}p(x)$};
      \node at (\x,0) {$\phantom{\freezeshort{x}\eventually_{(0,1]}}p(x)$};

      \renewcommand{\x}{95}
      \node at (\x,-25) {(a) for $w_0$};

      \draw (\x-20,100) -- node[above] {$[0,\infty)$} ++ (40,0);

      \node[node] (Nfreeze) at (\x,80) {$\emap$};
      \node[node] (Neventually) at (\x,40) {$\emap$};
      \node[node] (Natomic) at (\x,0) {$\emap$};
      
      \node[guard] (Geventually) at (\x,33) {};
      \node[guard] (Gfreeze) at (\x,73) {};
      
      \draw[trigger] (Natomic) to [out=90, in=-90] (Geventually);        
      \draw[trigger] (Neventually) to [out=90, in=-90] (Gfreeze);

      \renewcommand{\x}{160}
      \node at (\x+50,-25) {(b) for $w_1$};

      \draw (\x-10,100) -- node[above] {$[0,\tau_1)$} ++ (40,0);
      \draw[|-|] (\x+50,100) -- node[above] {$\set{\tau_1}$} ++ (0.1,0);
      \draw (\x+70,100) -- node[above] {$(\tau_1,\infty)$} ++ (40,0);
      
      \node[node] (NfreezeL) at (\x+10, 80) {$\emap$};
      \node[node] (NeventuallyL) at (\x+10, 40) {$\emap$};
      \node[node] (NatomicL) at (\x+10, 0) {$\emap$}; 
      
      \node[node] (NfreezeM) at (\x+50, 80) {$\emap$};
      \node[node] (NeventuallyM) at (\x+50, 40) {$\nu$};
      \node[node] (NatomicM) at (\x+50, 0) {$\emap$};
      
      \node[node] (NfreezeR) at (\x+90, 80) {$\emap$};
      \node[node] (NeventuallyR) at (\x+90, 40) {$\emap$};
      \node[node] (NatomicR2) at (\x+94, 4) {};
      \node[node, fill=white] (NatomicR) at (\x+90, 0) {$\emap$};

      \node[guard] (GeventuallyL1) at (\x+4, 33) {};
      \node[guard] (GeventuallyL2) at (\x+10, 33) {};
      \node[guard] (GeventuallyL3) at (\x+16, 33) {};
      \node[guard] (GfreezeL) at (\x+10, 73) {};
      
      \node[guard] (GeventuallyM) at (\x+50, 33) {};
      \node[guard] (GfreezeM) at (\x+50, 73) {};
      
      \node[guard] (GeventuallyR) at (\x+90, 33) {};
      \node[guard] (GfreezeR) at (\x+90, 73) {};
      
      \draw[trigger] (NatomicL) to [out=90, in=-90] (GeventuallyL1); 
      \draw[trigger] (NatomicM) .. controls (\x+50,20)  and (\x+10,0) .. (GeventuallyL2);
      \draw[trigger] (NatomicR) to [out=90, in=-90] (GeventuallyR);        
      \draw[trigger] (NatomicR2) to [out=90, in=-90] (GeventuallyM);
      \draw[trigger] (NatomicR) .. controls (\x+70,30) and (\x+16,5) .. (GeventuallyL3);

      \draw[trigger] (NeventuallyL) to [out=90, in=-90] (GfreezeL);        
      \draw[trigger] (NeventuallyM) to [out=90, in=-90] (GfreezeM);        
      \draw[trigger] (NeventuallyR) to [out=90, in=-90] (GfreezeR);        
      
      
      
    \end{tikzpicture}}
  \caption{Graph structures.}
  \label{fig:datastruct}
\end{figure}
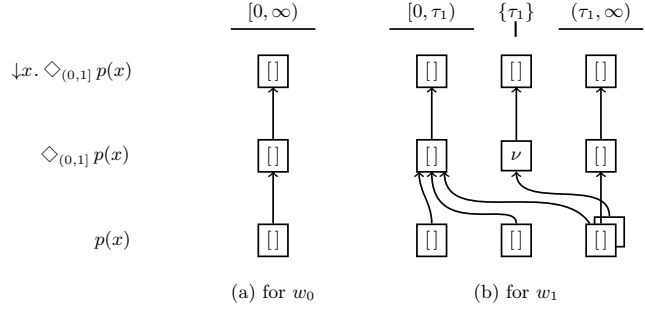
Figure~\ref{fig:datastruct} illustrates the data structure at the end
of iterations~0 and~1, that is, corresponding to the
observations~$w_0$ and~$w_1$, for the setting in
Example~\ref{ex:freeze}.
A box in the figure corresponds to a node of the graph structure,
where the node's formula is given by the row, the interval by the
column of the box, and the valuation by the content of the box. The
valuation of the hidden box in the lower right corner is
$\nu=[x\mapsto d]$, the same as the box in the middle of
Figure~\ref{fig:datastruct}(b).
Arrows correspond to triggers.

\subsubsection{Correctness.}

The following theorem establishes the monitor's correctness. 
We refer to Appendix~\ref{app:proof} for its proof.
\begin{theorem}
  \label{thm:correctness}
  Let $\bar{w}$ be the valid observation sequence derived from the
  messages received by {\normalfont\ls{Monitor}}.  Furthermore, let
  $\phi$ be a closed \MTLdata formula.
  {\normalfont\ls{Monitor($\phi$)}} is observationally complete and 
  sound for $\bar{w}$ and $\phi$.
\end{theorem}

An important property class in monitoring are safety properties. We
note that our monitor is not limited to formulas of this class, and
the monitor is sound and observationally complete for any formula.
For instance, for the formula $\phi = \always\eventually p$, which
states that $p$ is true infinitely often, the monitor will never
output a verdict, as expected. It is nevertheless observationally
complete for any valid observation sequence $\bar{w}$, since
$\eosem{w_i,\tau,\nu}{\phi}=\bot$, for any $i\in\Nat$, $\tau\in\Qpos$,
and partial valuation~$\nu$.

Besides correctness requirements, time and space requirements are
also important. Concerning time requirements, recall that the
underlying decision problem is PSPACE-complete, see
Theorem~\ref{thm:decidable}. Concerning space requirements, note
that space cannot be bounded even in the setting without message loss
and with in-order delivery. Indeed, consider the formula
$\always \freezeshort{x} p(x) \to \always_{(0,\infty)} \neg p(x)$
stating that the parameter of $p$ events are fresh at each time point.
Any monitor must store the parameters seen.
Further investigation of the time and space complexity of the
monitoring procedure is left for future work.

%% file: casestudy.tex
\section{Experiments}
\label{sec:casestudy}

We have implemented our monitor in a prototype tool, written in the
programming language Go.  Our tool either reads messages from a log
file or over a UDP socket.  Our experimental
evaluation focuses on the prototype's performance in settings with
different message orderings.

\paragraph{Setup.}

We monitor the formulas in
Figure~\ref{fig:specifications}, which vary in their temporal
requirements and the data involved. They express compliance policies
from the banking domain and are variants of policies that have been
used in previous case studies~\cite{Basin_etal:rv_mfotl}.
\begin{figure}[t]
  \centering
  \footnotesize
  \scalebox{.8}{
    \begin{minipage}{15.2cm}
  \begin{gather}
    \hspace{-.5cm}
    \tag{P1}
    \always 
    \freeze{cid}{c} \freeze{tid}{t} \freeze{sum}{a} 
    \mathit{trans}(c,t,a)\wedge a>2000 
    \rightarrow \eventually_{[0,3]} \mathit{report}(t)
    \\
    \hspace{-.5cm}
    \tag{P2}
    \always 
    \freeze{cid}{c} \freeze{tid}{t} \freeze{sum}{a}{}
    \mathit{trans}(c,t,a) \land a>2000\to
    \always_{(0,3]}\freeze{tid}{t'} \freeze{sum}{a'} 
    \mathit{trans}(c,t',a') \rightarrow a'\leq 2000
    \\
    \hspace{-.5cm}
    \tag{P3}
    \always 
    \freeze{cid}{c}{} \freeze{tid}{t} \freeze{sum}{a}{}
    \mathit{trans}(c,t,a) \land a>2000\to
    \big((\freeze{tid}{t'} \freeze{sum}{a'} \mathit{trans}(c,t',a')\rightarrow t=t')
    \weakuntil \mathit{report}(t)\big)
    \\
    \hspace{-.5cm}
    \tag{P4}
    \always 
    \freeze{cid}{c}{} \freeze{tid}{t} \freeze{sum}{a}{}
    \mathit{trans}(c,t,a) \land a>2000\to
    \always_{[0,6]}
    \freeze{tid}{t'} \freeze{sum}{a'} \mathit{trans}(c,t',a')
    \to \eventually_{[0,3]} \mathit{report}(t') 
  \end{gather}
    \end{minipage}}
  \caption{\MTLdata formulas used in the experimental evaluation.}
  \label{fig:specifications}
\end{figure}
Furthermore, we synthetically generate log files.  Each log spans
over 60 time units (i.e., a minute) and contains one event per time
point.  The number of events in a log is determined by the
\emph{event rate}, which is the approximate number of events per time
unit (i.e., a second).  For each time point~$i$, with $0\leq i<60$, the
number of events with a timestamp in the time interval $[i,i+1)$ is
randomly chosen within ±10\% of the event rate.  The events and their
parameters are randomly chosen such that the number of violations is
in a provided range.  For instance, a log with event rate 100
comprises approximately 6000 events.  Finally, we use a standard
desktop computer with a 2.8Ghz Intel Core i7 CPU, 8GB of RAM, and the
Linux operating system.

\begin{figure}[t]
  \centering
  \subfigure[in-order]{\includegraphics[scale=.20]{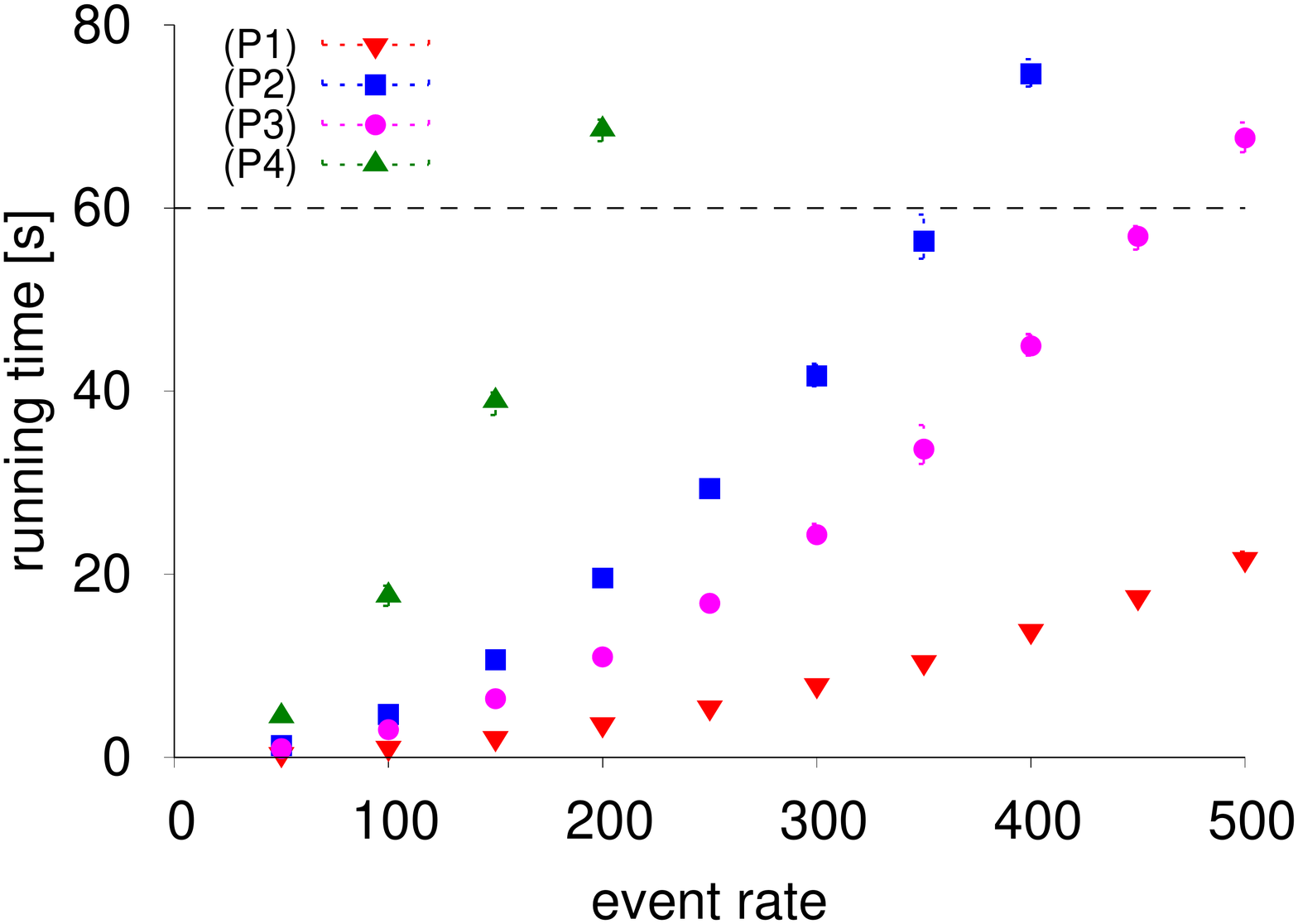}}
  \subfigure[out-of-order]{\includegraphics[scale=.20]{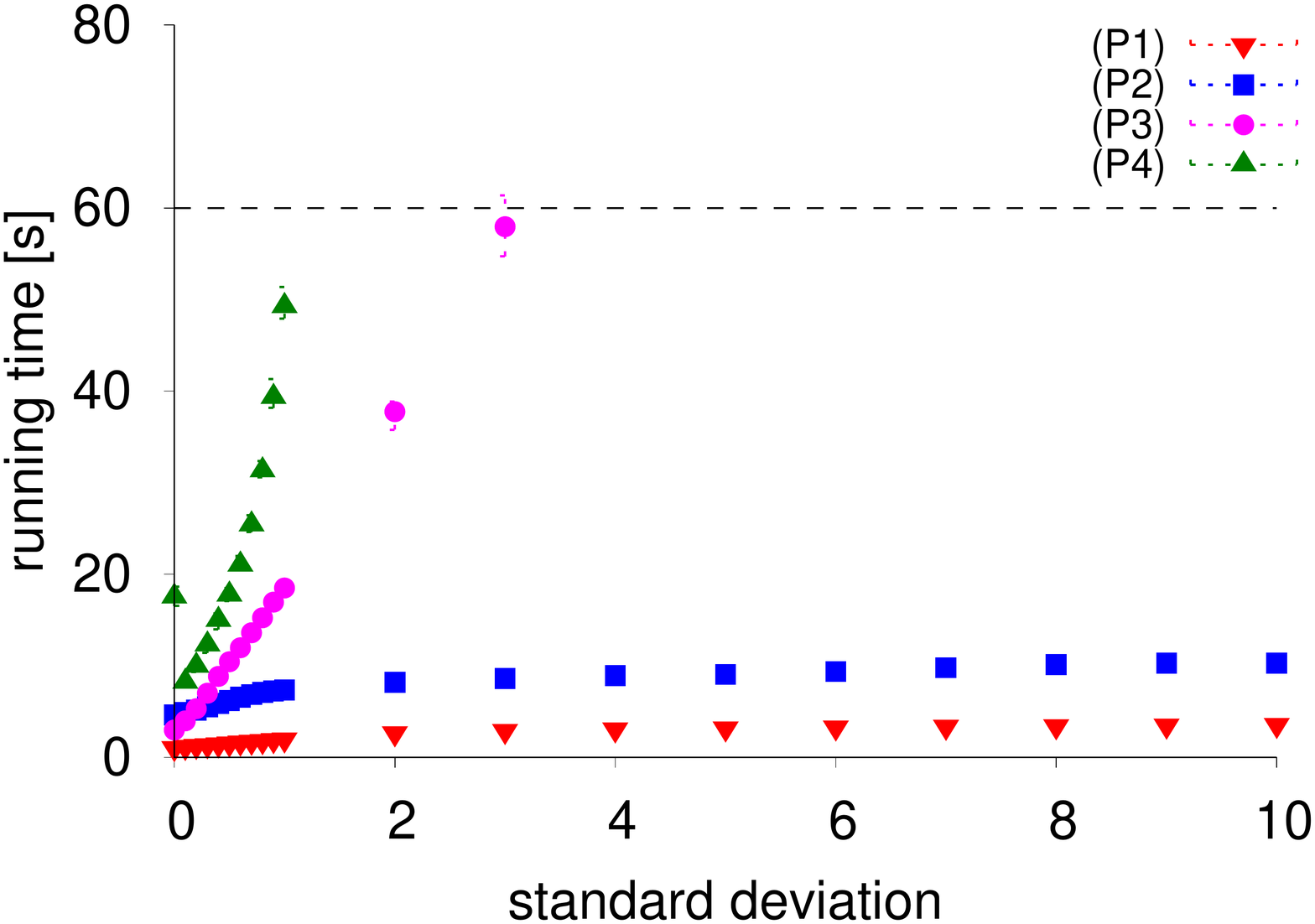}}
  \caption{Running times (where each data point shows the mean over five
    logs together with the minimum and maximum, which are very close
    to the mean).}
  \label{fig:performance}
\end{figure}

\paragraph{In-order Delivery.}

In our first setting, messages are received ordered by their
timestamps and are never lost. Namely, all events of the log are
processed in the order of their timestamps.
Figure~\ref{fig:performance}(a) shows the running times of our
prototype tool for different event rates.  Note that each log spans 60
seconds and a running time below 60 seconds essentially means that the
events in the log could have been processed online.  
The dashed horizontal lines mark this border.

\paragraph{Out-of-order Delivery.}

In our second setting, messages can arrive out of order but they are
not lost. We control the degree of message arrival disruption as
follows.
For the events in a generated log file, we choose their \emph{arrival
  times}, which provide the order in which the monitor processes them.
The arrival time of an event is derived from the event's timestamp by
offsetting it by a random delay with respect to the normal
distribution with a mean of 10 time units and a chosen standard
deviation.  In particular, for an event's timestamp~$\tau$ and for a
standard deviation~$\sigma>0$, it holds that an arrival time $\tau'$
is in the interval $[\tau+10-\sigma,\tau+10+\sigma]$ with
probability~$0.68$ and in $[\tau+10-2\sigma,\tau+10+2\sigma]$ with
probability~$0.95$.  For the degenerate case $\sigma=0$, the reordered
log is identical to the original log.
We remark that the mean value does not impact the event reordering
because it does not influence the difference between arrival times.

Figure~\ref{fig:performance}(b) shows the prototype's running times on
logs with the fixed event rate 100 for different deviations.  For
instance, for (P1), the logs are processed in around $1$~second when
$\sigma=0$ and in $3.5$~seconds when $\sigma=10$.

\paragraph{Interpretation.}

The running times are nonlinear in the event rate for all four
formulas.  This is expected from Theorem~\ref{thm:decidable}.  The
growth is caused by the data values occurring in the events.  A log
with a higher event rate contains more different data values and the
monitor's state must account for those.
As expected, (P1) is the easiest to monitor. It has only one block of
freeze quantifiers.  Note that (P1)--(P3) have two temporal
connectives, where one is the outermost connective~$\always$, which is
common to all formulas, whereas (P4) has an additional nesting of
temporal connectives. The time window is also larger than in~(P1)
and~(P2).

Also expected, the running times increase when messages are received
out of order.  Again, (P4) is worst. For (P1) and (P2), however, the
growth rate decreases for larger standard deviations. This is because,
as the standard deviation increases, all the events within the
relevant time window for a given time point arrive at the monitor in
an order that is increasingly close to the uniformly random one. The
running times thus stabilize.  Due to the larger time window, this
effect does not take place for (P4).  The running times for (P3)
increase more rapidly than for (P1) and (P2) because of the data
values and the ``continuation formula'' of the derived unbounded
temporal connective $\weakuntil$.

To put the experimental results in perspective, we carried out two
additional experiments.
First, we conducted similar experiments on formulas with their freeze
quantifiers removed and further transformed into propositional
formulas, as described in Appendix~\ref{app:propositonal}.  We make
similar observations in the propositional setting. However, in the
propositional setting the running times increase linearly with respect
to the event rate and logs are processed several orders of magnitude
faster.  Overall, one pays a price at runtime for the expressivity
gain given by the freeze quantifier.
Second, we compared our prototype with the MONPOLY tool~\cite{monpoly}.
MONPOLY's specification language is, like \MTLdata, a point-based
real-time logic.  It is richer than \MTLdata in that it admits
existential and universal quantification over domain
elements. However, MONPOLY specifications are syntactically restricted
in that temporal future connectives must be bounded (except for the
outermost connective $\always$).
Thus, (P3) does not have a counterpart in MONPOLY's specification
language.
MONPOLY handles the counterparts of (P1), (P2), and (P4) significantly
faster, up to three orders of magnitude.  Comparing the performance of
both tools should, however, be taken with a grain of salt.  First,
MONPOLY only handles the restrictive setting where messages must be
received in-order.  Second, MONPOLY outputs violations for
specifications with (bounded) future only after all events in the
relevant time window are available, whereas our prototype outputs
verdicts promptly.\footnote{For instance, for the formula
  $\always_{[0,3]} p$, if $p$ does not hold at time point $i$ with
  timestamp~$\tau$, then our prototype outputs the corresponding
  verdict directly after processing the time point $i$, whereas MONPOLY
  reports this violation at the first time point with a timestamp
  larger than $\tau+3$.}  Finally, while MONPOLY is optimized, our
prototype is not.

In summary, our experimental evaluation shows that one pays a high
price to handle an expressive specification language together with
message delays. Nevertheless, our prototype's performance is
sufficient to monitor systems that generate hundreds of events per
second, and the prototype can be used as a starting point for a more
efficient implementation.


%% file: related.tex
\section{Related Work}
\label{sec:related}

Runtime verification is a well-established approach for checking at
runtime whether a system's execution fulfills a given specification.
Various monitoring algorithms exist,
e.g.,~\cite{Barringer_etal:eagle,Bauer_etal:rv_tltl,Meredith_etal:mop,Basin_etal:rv_mfotl}.
They differ in the specifications they can handle (some of the
specification languages account for data values) and they make
different assumptions on the monitored systems.  A commonly made
assumption is that a monitor has always complete knowledge about the
system behavior up to the current time.  Only a few
runtime-verification approaches exist that relax this assumption.
Note that this assumption is, for instance, not met in distributed
systems whose components communicate over unreliable channels.

Closest to our work is the runtime-verification approach
by Basin et al.~\cite{Basin_etal:failureaware_rv}.
We use the same system model
and our monitoring algorithm extends their monitoring algorithm for
the propositional real-time logic MTL. Namely, our algorithm handles
the more expressive specification language \MTLdata and handles data values.
Furthermore, we present a semantics for \MTLdata that is based on
three truth values and uses observations instead of timed words.  This
enables us to cleanly state correctness requirements and establish
stronger correctness guarantees for the monitoring algorithm.  Basin
et al.'s completeness result~\cite{Basin_etal:failureaware_rv} is
limited in that it assumes that all messages are eventually received.
Finally, Basin et al.~\cite{Basin_etal:failureaware_rv} do not
evaluate their monitoring algorithm experimentally.

Colombo and Falcone~\cite{ColomboFalcone:global_clock} propose 
a runtime-verification approach, based on formula rewriting, that 
also allows the monitor to receive messages out of order. 
Their approach only handles the propositional
temporal logic LTL with the three-valued semantics proposed by 
Bauer et al.~\cite{Bauer_etal:ltl_rv}.
In a nutshell, their approach unfolds
temporal connectives as time progresses and 
special propositions act as placeholders for subformulas.
The subsequent assignment of these placeholders to Boolean truth values 
triggers the reevaluation and simplification of the formula.
Their approach only guarantees soundness but not completeness,
since the simplification rules used for formula rewriting are incomplete.
Finally, its performance with respect to out-of-order 
messages is not evaluated.

The monitoring approaches by Garg et
al.~\cite{Garg_etal:policy_incomplete} and Basin et
al.~\cite{BasinKMZ-RV12}, both targeting the auditing of policies on
system logs, also account for knowledge gaps, i.e., logs that may not
contain all the actions performed by a system.  Both
approaches handle rich policy specification languages with first-order
quantification and a three-valued semantics.
Garg et al.'s approach~\cite{Garg_etal:policy_incomplete}, which is
based on formula rewriting, is however, not suited for online use
since it does not process logs incrementally.  It also only accounts
for knowledge gaps in a limited way, namely, the interpretation of a
predicate symbol cannot be partially unknown, e.g., for certain time
periods. Furthermore, their approach is not complete.
Basin et al.'s approach~\cite{BasinKMZ-RV12}, which is based on their
prior work~\cite{Basin_etal:rv_mfotl}, can be used online.  However,
the problem of how to incrementally output verdicts as prior knowledge
gaps are resolved is not addressed, and thus it does not deal with
out-of-order events.  Moreover, the semantics of the specification
language handled does not reflect a monitor's partial view about the
system behavior. Instead, it is given for infinite data streams that
represent system behavior in the limit.

Several dedicated monitoring approaches for distributed systems have
been developed~\cite{Sen_etal:decentralized_disitributed_monitoring,Mostafa_Bonakdarbour:decentralized_rv,Bauer_Falcone:decentralised_monitor}.
These approaches only handle less expressive specification languages,
namely, the propositional temporal logic LTL or variants thereof.
Furthermore, none of them handles message loss or out-of-order
delivery of messages, problems that are inherent to such systems
because of crashing components and nonuniform delays in message
delivery.

A similar extension of MTL with the freeze quantifier is defined by
Feng et al.~\cite{Feng_etal:MTL_data}.  Their analysis focuses on the
computational complexity of the path-checking problem. However, they
use a finite trace semantics, which is less suitable for runtime
verification. Out-of-order messages are also not considered.

Temporal logics with additional truth values have also been considered
in model checking finite-state systems.  Closest to our three-valued
semantics is the three-valued semantics for LTL by Goidefroid and
Piterman~\cite{Goidefroid_Piterman:gmc}, which is based on infinite
words, not observations (Definition~\ref{def:observation}).  Similar
to (T\ref{enum:observation_data}) of Definition~\ref{def:observation},
a proposition with the truth value~$\unknown$ at a position can be
refined by $\true$ or $\false$.  In contrast, their semantics does not
support refinements that add and delete letters,
cf.~(T\ref{enum:observation_split})
and~(T\ref{enum:observation_removal}) of
Definition~\ref{def:observation}.


%% file: concl.tex
\section{Conclusion}
\label{sec:concl}

We have presented a runtime-verification approach to checking
real-time specifications given as \MTLdata formulas.  Our approach
handles the practically-relevant setting where messages sent to the
monitors can be delayed or lost, and it provides soundness and
completeness guarantees.
Although our experimental evaluation is promising, 
our approach does not yet scale to monitor systems that generate
thousands or even millions of events per second.  This requires
additional research, including algorithmic optimizations.  
We plan to do this in future work, as well as 
to deploy and evaluate our approach in
realistic, large-scale case studies.

\paragraph{Acknowledgments.} 

This work was partly performed within the 5G-ENSURE project
(www.5gensure.eu) and received funding from the EU Framework Programme
for Research and Innovation Horizon 2020 under grant agreement
no.~671562.  David Basin acknowledges support from the Swiss National
Science Foundation grant Big Data Monitoring (167162).


%% file: app_mtl.tex
\section{Additional \MTLdata Details}
\label{app:mtl_bool}

\subsection{Standard Boolean Semantics}
\label{subapp:synsem}

\MTLdata with a Boolean semantics has the same syntax as \MTLdata with
a three-valued semantics as defined in Section~\ref{sec:mtl}.  We use
the same syntactic sugar and conventions from Section~\ref{sec:mtl}.
The Boolean semantics for \MTLdata is defined over timed words.

To define \MTLdata's Boolean semantics, we introduce the following
notation and terminology.
Recall that
$D$---the \emph{data domain}---is a nonempty set of values.
Furthermore, let $\Sigma$ be the set of the pairs $(\sigma,\rho)$,
where $\sigma$ is a total function over $P$ with
$\sigma(p)\subseteq D^{\iota(p)}$ for $p\in P$ and $\rho$ is a total
function over $R$ with $\rho(r)\in D$ for $r\in R$.  In the following,
timed words are always over the alphabet $\Sigma$.  Note that the
letter at position $i\in\Nat$ of a timed word (over $\Sigma$) is of
the form $(\tau_i, \sigma_i,\rho_i)$, where $\sigma_i$ interprets the
predicate symbols at time $\tau_i$ and $\rho_i$ determines the values
stored in the registers in $R$ at time $\tau_i$.
As for observations, we call the positions $i\in\Nat$ of
a timed word $w$ the \emph{time points} of $w$. Furthermore, we call
$\tau_i$ the \emph{timestamp} of the time point $i\in\Nat$.

\begin{remark}
  Observations generalize the notion of a finite prefix of a timed
  word. Let $w=(\tau_1,\sigma_1,\rho_1) (\tau_2,\sigma_2,\rho_2)\dots$
  be a timed word. For the prefix of length $n\in\Nat$ of $w$, we
  define the word $w_n$ as $(\set{\tau_1},\sigma_1,\rho_1) \dots
  (\set{\tau_n},\sigma_n,\rho_n)(I,[\,],[\,])$, with
  $I=\Qpos\setminus\big(\bigcup_{1\leq i\leq n}[0,\tau_i]\big)$.  That
  is, we transform the timestamps of the prefix into singletons and
  attach a last letter, which can be seen as a placeholder for the
  remaining letters in $w$.
  The words~$w_n$ for $n\in\Nat$ are observations.  Obviously,
  $w_0=([0,\infty),[\,],[\,])$ is an observation. For $n>0$, we obtain
  $w_n$ from $w_{n-1}$ by applying the transformation
  (T\ref{enum:observation_split}) with the timestamp $\tau_n$ on
  $w_{n-1}$'s last letter, then applying
  (T\ref{enum:observation_removal}) to delete the letter with the
  interval $(\tau_{n-1},\tau_n)$, and finally applying
  (T\ref{enum:observation_data}) on the letter with interval
  $\set{\tau_n}$ to populate it with $\sigma_n$ and $\rho_n$.
\end{remark}

\MTLdata's Boolean semantics is defined inductively over the formula
structure.  In particular, similar to the function
$\phi\mapsto \osem{w,i,\nu}{\phi}$ defined in Section~\ref{sec:mtl},
we define a function $\phi\mapsto\sem{w,i,\nu}{p(\bar{x})}\in\Two$,
for a given timed word~$w$, a time point~$i\in\N$, and a
valuation~$\nu:V\pto D$.  Let
$w=(\tau_0,\sigma_0,\rho_0) (\tau_1,\sigma_1,\rho_1)\dots$.
\begin{equation*}
  \begin{array}{@{}rcl@{}}
    \sem{w,i,\nu}{\true}
    &:=&
    \true
    \\[.1cm] 
    \sem{w,i,\nu}{p(\bar{x})}
    &:=&
    \begin{cases}
      \true & \text{if $\nu(\bar{x})\in\sigma_i(p)$}
      \\
      \false & \text{otherwise}
    \end{cases}
    \\[.4cm] 
    \sem{w,i,\nu}{\freeze{r}{x}\phi}
    &:=&
    \sem{w,i,\nu[x\mapsto \rho_i(r)]}{\phi}
    \\[.1cm] 
    \sem{w,i,\nu}{\neg\phi}
    &:=&
    \neg\sem{w,i,\nu}{\phi}
    \\[.1cm] 
    \sem{w,i,\nu}{\phi\vee\psi}
    &:=&
    \sem{w,i,\nu}{\phi}
    \vee
    \sem{w,i,\nu}{\psi}
    \\[.1cm] 
    \sem{w,i,\nu}{\phi\until_I\psi}
    &:=&
    \bigvee_{j\in \setx{\ell\in\N}{\tau_\ell-\tau_i\in I}}
      \big(\sem{w,j,\nu}{\psi} 
        \wedge\bigwedge_{i\leq k<j}\sem{w,k,\nu}{\phi}\big)
  \end{array}
\end{equation*}
Note that we abuse notation here and identify the logic's constant
symbol~$\true$ with the Boolean value $\true$, and the
connectives~$\neg$ and $\vee$ with the corresponding logical
operators.

\begin{definition}
For $\tau\in\Qpos$, a timed word~$w$, a valuation $\nu$, and a formula
$\phi$, we define $\esem{w,\tau,\nu}{\phi}:=\sem{w,j,\nu}{\phi}$,
provided that there is some time point $j$ in $w$ with timestamp
$\tau$, and $\esem{w,\tau,\nu}{\phi}:=\unknown$, otherwise.
\end{definition}

The following theorem shows that the semantics of \MTLdata on
observations conservatively extends the logic's semantics on timed
words. In particular, if a formula evaluates to a Boolean value for an
observation for a given time $\tau\in\Qpos$, it has the same Boolean
value on any timed word that refines the observation.
\begin{theorem}
  \label{thm:osem_refinement}
  Let $\phi$ be a formula, $\mu$ a partial valuation, $\nu$ a total
  valuation, $u$ an observation, $v$ a timed word, and
  $\tau\in\Qpos$. If $u\sqsubseteq v$ and $\mu\sqsubseteq \nu$ then
  $\eosem{u,\tau,\mu}{\phi} \preceq \esem{v,\tau,\nu}{\phi}$.
\end{theorem}
\begin{proof} 
  Let $(I_i,\sigma_i,\rho_i)$ and $(\tau_j,\sigma'_j,\rho'_j)$, for
  $i\in\position(u)$ and $j\in\Nat$ be the letters of $u$ and~$v$,
  respectively.
  As $u\sqsubseteq v$, we have that there is a function
  $\pi:\Nat\to\position(u)$ such that (R1)~$\tau_j\in I_{\pi(j)}$,
  (R2)~$\sigma_{\pi(j)}\sqsubseteq \sigma'_j$, and
  (R3)~$\rho_{\pi(j)}\sqsubseteq\rho'_j$, for every $j\in\Nat$.
  It is easy to see that $\pi$ is monotonic.

  We prove by structural induction on $\phi$ that for any time point
  $i'\in\Nat$ and partial valuations~$\mu$ and~$\nu$ with
  $\mu\sqsubseteq\nu$, it holds that
  $\osem{u,\pi(i'),\mu}{\phi} \preceq \sem{v,i',\nu}{\phi}$.  The
  theorem's statement easily follows from this property.
  For the reminder of the proof, we fix an arbitrary time point
  $i'\in\Nat$ and arbitrary partial valuations~$\mu$ and $\nu$ with
  $\mu\sqsubseteq\nu$.
  Let $i=\pi(i')$.

  When $\osem{u,i,\mu}{\phi}=\bot$ the statement clearly holds.
  Hence, it suffices to show that
  $\osem{u,i,\mu}{\phi} = \sem{v,i',\nu}{\phi}$, provided that
  $\osem{u,i,\mu}{\phi}\in\Two$.

  \emph{Base cases.} The case $\phi=\true$ is trivial. Consider the
  case $\phi=p(\bar x)$, for some $p\in P$. As
  $\osem{u,i,\nu}{p(\bar{x})}\in\Two$, it holds that
  $\bar{x}\in\pdef(\nu)$ and $p\in\pdef(\sigma_i)$. It follows
  from the theorem's premise that $\mu(\bar{x})=\nu(\bar{x})$, and
  from~(R2) that $\sigma_{i}(p)=\sigma'_{i'}(p)$. Thus
  $\osem{u,i,\nu}{p(\bar x)} = \sem{v,i',\nu}{p(\bar x)}$.

  \emph{Inductive cases.} The cases where $\phi$ is of the form
  $\neg\alpha$ or $\alpha\vee\beta$ 
  are straightforward and omitted. The remaining cases are as follows.

  First, assume that $\phi$ is of the form $\freeze{r}{x}\psi$.
  Let $\eta=\mu[x\mapsto \rho_i(r)]$ and
  $\eta'=\nu[x\mapsto \rho'_{i'}(r)]$.
  By~(R3), we have that if $r\in\pdef(\rho_i)$ then
  $r\in\pdef(\rho'_{i'})$ and $\rho_i(r)=\rho'_{i'}(r)$, and thus
  $\eta(x)=\eta'(x)$. Furthermore, if $r\not\in\pdef(\rho_i)$, then
  $x\not\in\pdef(\eta)$. This shows that $\eta\sqsubseteq\eta'$.
  It follows from the induction hypothesis that
  $\osem{u,i,\eta}{\psi} \preceq \sem{v,i',\eta'}{\psi}$. We conclude
  that $\osem{u,i,\mu}{\phi}=\sem{v,i',\nu}{\phi}$.

  Finally, assume that $\phi$ is of the form $\alpha\until_I\beta$. We
  consider first the case
  $\osem{u,i,\mu}{\alpha\until_I\beta}=\true$. By definition, there is
  a $j\geq i$ such that $\istp_u(j)=\true$, $\tc_{u,I}(i,j)=\true$,
  $\osem{u,j,\mu}{\beta}=\true$, and
  $\istp_u(k)\to\osem{u,k,\mu}{\alpha}=\true$, for all $k$ with
  $i \leq k < j$.
  As $j$ is a time point in $u$, there is a time point $j'$ in $v$
  such that $\pi(j')=j$. From~(R1) we have that
  $\tau_{j'}=\timestamp_{u}(j)$.  As $\tc_{u,I}(i,j)=\true$ and
  $I_j=\set{\tau_{j'}}$, we have that $\tau_{j'}-\tau\in I$, for all
  $\tau\in I_i$. From (R1), we have that $\tau_{i'}\in I_i$. Thus
  $\tau_{j'}-\tau_{i'}\in I$~(I1).
  From the induction hypothesis, we have that
  $\osem{u,j,\mu}{\beta} \preceq \sem{v,j',\nu}{\beta}$. Hence
  $\sem{v,j',\nu}{\beta}=\true$~(I2).
  From the induction hypothesis, we also have that 
  $\osem{u,\pi(k'),\mu}{\alpha}\preceq\sem{v,k',\nu}{\alpha}$, for
  any $k'\in\Nat$.
  Let $k'\in\Nat$ such that $i'\leq k'< j'$, and let $k=\pi(k')$. From
  the monotonicity of~$\pi$ we have that $i\leq k\leq j$. Since $j$ is
  a time point in $u$ we also have that $k<j$.
  As $\istp_u(k)\to\osem{u,k,\mu}{\alpha}=\true$ and $\istp$ is never
  $\false$ by definition, we have that
  $\osem{u,k,\mu}{\alpha}=\true$. Then
  $\sem{v,k',\mu}{\alpha}=\true$~(I3).
  Summing up, from (I1), (I2), (I3), and as $k'$ was chosen
  arbitrarily, we obtain that
  $\sem{v,i',\nu}{\alpha\until_I\beta}=\true$.
  
  The case $\osem{u,i,\mu}{\alpha\until_I\beta} = \false$ is as
  follows.  Note that each disjunct in the definition of
  $\osem{u,i,\mu}{\alpha\until_I\beta}$ is $\false$.
  We fix an arbitrary $j'\geq i'$ and let $j=\pi(j')$. It holds that
  $\istp_u(j) \wedge \tc_{u,I}(i,j) \wedge \osem{u,j,\mu}{\beta}
  \wedge \bigwedge_{i\leq k<j} (\istp_u(k)\rightarrow
  \osem{u,k,\mu}{\alpha}) = \false$.
  Since $\istp_u(j)\not=\false$, one of the remaining conjuncts must
  be $\false$.
  \begin{enumerate}[(1)]
  \item If $\tc_{u,I}(i,j)=\false$, then $\tau'-\tau''\not\in I$, for
    all $\tau''\in I_i$ and $\tau'\in I_j$. From~(R1),
    $\tau_{i'}\in I_i$ and $\tau_{j'}\in I_j$. It follows that
    $\tau_{j'}-\tau_{i'}\not\in I$.
  \item If $\osem{u,j,\mu}{\beta}=\false$, then
    $\sem{v,j',\nu}{\beta}=\false$, by induction hypothesis.
  \item If $\istp_u(k)\to\osem{u,k,\mu}{\alpha}=\false$, for some $k$
    with $i\leq k< j$, then $\istp_u(k)=\true$ and
    $\osem{u,k,\mu}{\alpha}=\false$. It follows as before that there
    is a $k'$ with $i'\leq k'<j'$ such that
    $\sem{v,k',\nu}{\alpha}=\false$.
  \end{enumerate}
  We have thus obtained that either $\tau_{i'}-\tau_{j'}\not\in I$ or
  one of the conjuncts of
  $\sem{v,j',\nu}{\beta} \wedge \bigwedge_{i'\leq k'< j'}
  \sem{v,k',\nu}{\alpha}$ is $\false$. In other words, if $j'$ is
  such that $\tau_{j'}-\tau_{i'}\in I$, then
  $\sem{v,j',\nu}{\beta} \wedge \bigwedge_{i'\leq k'<j'}
  \sem{v,k',\nu}{\alpha} = \false$.
  As $j'$ was chosen arbitrarily, we conclude that
  $\sem{v,i',\nu}{\alpha\until_I\beta}=\false$.  \qed
\end{proof}

\subsection{Correctness Requirements}
\label{subapp:req_bool}

In this section, we formulate similar monitoring requirements as in
Section~\ref{subsec:requirements}, formulating them this time with
respect to the Boolean \MTLdata semantics. We then argue that these
requirements are too strong.

For an observation~$w$, we define $U_w:=\setx{v}{v\text{ a timed word
    with }w\sqsubseteq v}$. Intuitively, $U_w$ contains all the timed
words that are compatible with the reported system behavior that a
monitor received so far, represented by $w$.
\begin{definition}
  \label{def:soundness_completeness}
  Let $M$ be a monitor, $\phi$ a formula, and $\bar{w}$ a valid
  observation sequence.
  \begin{itemize}[--]
  \item $M$ is \emph{sound} for $\bar{w}$ and $\phi$ if for all
    valuations~$\nu$ and $i\in\Nat$, whenever $(\tau,b)\in M(w_i)$
    then  $\bigcurlywedge_{v\in U_{w_i}}\esem{v,\tau,\nu}{\phi}=b$.

  \item $M$ is \emph{complete} for $\bar{w}$ and $\phi$ if for all
    valuations~$\nu$, $i\in\Nat$, and $\tau\in\Q_{\geq0}$, whenever
    $\bigcurlywedge_{v\in U_{w_i}}\esem{v,\tau,\nu}{\phi}\in\Two$
    then
    $(\tau,b)\in \bigcup_{j\leq i}M(w_j)$, for some $b\in\Two$.
  \end{itemize}
  We say that $M$ is \emph{sound} if $M$ is sound for all valid
  observation sequences $\bar{w}$ and formulas $\phi$. The definition
  of $M$ being \emph{complete} is analogous.
\end{definition}

The correctness requirements in
Definition~\ref{def:soundness_completeness} are related to
the use of a three-valued ``runtime-verification'' semantics for a
specification language, as introduced by Bauer et
al.~\cite{Bauer_etal:rv_tltl} for LTL and adopted by other
runtime-verification approaches (e.g.~\cite{Bauer-FMSD15}).
Intuitively speaking, both a sound and complete monitor and a monitor
implementing a three-valued ``runtime-verification'' semantics output
a verdict as soon as the specification has the same Boolean value on
every extension of the monitor's current knowledge.  However, Bauer et
al.~\cite{Bauer_etal:rv_tltl} make no distinction between a monitor's
soundness and its completeness. Distinguishing these two
requirements separates concerns and this is important, as explained
next.
Ideally, a monitor is both sound and complete.  However, achieving
both of these properties can be hard or even impossible for non-trivial
specification languages, as we now explain, when relying on the
standard Boolean semantics.
\begin{remark}
  \label{rem:correctness}
  For a specification language, having a sound and complete monitor
  $M$ for a specification language is at least as hard as checking
  satisfiability for this language.  For \MTLdata, a closed formula
  $\phi$ is satisfiable iff $(0,\false)\not\in M(w_1)$, assuming that
  $0$ is always the timestamp of the first time point of a timed word
  and the observation $w_1$ is obtained from the monitor's initial
  knowledge $w_0$ by (T\ref{enum:observation_split}) for the
  timestamp~$0$.  Note that $w_0=([0,\infty),[\,],[\,])$, $U_{w_0}$ is
  the set of all timed words, and $U_{w_1}=U_{w_0}$.  The
  propositional fragment of \MTLdata is already
  undecidable~\cite{OuaknineW06}.  There are fragments that are
  decidable but the complexity is usually high. Recall that for LTL,
  checking satisfiability is already PSPACE-complete.\footnote{In
    contrast to model checking, runtime verification is often
    advertised as a ``light-weight'' verification technique. In terms
    of complexity classes, the problem of soundly and completely
    monitoring finite-state systems with respect to LTL specifications
    is however at least as hard as the corresponding model-checking
    problem, which is PSPACE-complete.}

  Some monitoring approaches try to compensate for this complexity
  burden with a pre-processing step. For instance, the monitoring
  approach of Bauer et al.~\cite{Bauer_etal:rv_tltl} translates an LTL
  formula into an automaton prior to monitoring.  The resulting
  automaton can be directly used for sound and complete monitoring in
  environments where messages are neither delayed nor lost.  However,
  there are no obvious extensions for handling out-of-order message
  delivery.  Furthermore, not every specification language has such a
  corresponding automaton model and, for the ones where translations
  are known, the automaton construction can be very costly.  For LTL,
  the size of the automaton is already in the worst case doubly
  exponential in the size of the formula~\cite{Bauer_etal:rv_tltl}.
\end{remark}
In contrast the correctness requirements from
Definition~\ref{def:soundness_completeness-observation} are weaker and
achievable. This is due to the three-valued semantics, based on Kleene logic,  
for \MTLdata over observations. Furthermore, note that the three-valued semantics
for \MTLdata conservatively extends the standard Boolean semantics, as
shown by Theorem~\ref{thm:osem_refinement}.

Theorem~\ref{thm:osem_refinement} allows us to prove that the
completeness requirement from
Definition~\ref{def:soundness_completeness-observation} is indeed a
weaker notion than completeness requirement from
Definition~\ref{def:soundness_completeness}, while the soundness
requirement from Definition~\ref{def:soundness_completeness} offers
the same correctness guarantees as the one from
Definition~\ref{def:soundness_completeness-observation}.
\begin{theorem}
  \label{thm:obs_requirements}
  Let $M$ be a monitor. If $M$ is observationally sound, then $M$ is
  sound. If $M$ is complete, then $M$ is observationally complete.
\end{theorem}
\begin{proof}
  First, let $M$ be an observationally sound monitor. Let $\phi$ be a
  formula and $\bar w$ a valid observation sequence. Furthermore, let
  $\nu$ be a total valuation, $i\in\Nat$, $\tau\in\Qpos$, and
  $b\in\Two$ such that $(\tau,b)\in M(i)$. Then, by definition,
  $\eosem{w_i,\tau,\nu}{\phi}=b$. Now, let $v\in U_{w_i}$. We have
  that $w_i\sqsubseteq v$.  Then from Theorem~\ref{thm:osem_refinement}
  we have that $\esem{v,\tau,\nu}{\phi}=b$. As $v$ was chosen
  arbitrarily, we get
  $\bigcurlywedge_{v\in U_{w_i}}\esem{v,\tau,\nu}{\phi}=b$.
  Thus $M$ is a sound monitor.

  Now, let $M$ be a complete monitor.  Let $\phi$ be a formula and
  $\bar w$ a valid observation sequence. Furthermore, let $\nu$ be a
  partial valuation, $i\in\Nat$, $\tau\in\Qpos$ such that
  $\esem{w_i,\tau,\nu}{\phi}=b'$ for some $b'\in\Two$. 
  Let $\nu'$ be a total valuation with $\nu\sqsubseteq \nu'$.  Also,
  let $v\in U_{w_i}$. Then, as $w_i\sqsubseteq v$, from
  Theorem~\ref{thm:osem_refinement} we obtain that
  $\esem{v,\tau',\nu}{\phi}=b'$. As $v$ was chosen arbitrarily, we get
  $\bigcurlywedge_{v\in U_{w_i}}\esem{v,\tau,\nu'}{\phi}=b'$. As $M$
  is complete, it follows, by definition, that there is $b\in\Two$ and
  $j\leq i$ such that $(\tau,b)\in M(w_j)$. Thus $M$ is an
  observationally complete monitor.
  \qed
\end{proof}

\subsection{Proof of Theorem~\ref{thm:monotonicity}}

The following lemma characterizes the $\sqsubseteq$ relation on
observations.
\begin{lemma}
  \label{lem:refinement}
  Let $w$ and $w'$ be observations with letters
  $(I_i,\sigma_i,\rho_i)$ and respectively $(I'_j,\sigma'_j,\rho'_j)$,
  for $i\in\position(w)$ and $j\in\position(w')$. 
  If $w\sqsubseteq w'$ then there is a monotonic function
  $\pi:\position(w')\rightarrow \position(w)$ with the following properties.
  \begin{enumerate}[(R1)]
  \item\label{R:intervals} $I_j'\subseteq I_{\pi(j)}$, for all $j\in \position(w')$.
  \item\label{R:sigmas} $\sigma_{\pi(j)}\sqsubseteq \sigma_j'(p)$, for
    all $j\in \position(w')$ and $p\in P$.
  \item\label{R:rhos} $\rho_{\pi(j)}\sqsubseteq \rho'_{j}$, for all
    $j\in \position(w')$ and $r\in R$.
  \end{enumerate}
\end{lemma}
\begin{proof}[sketch]
  If $w=w'$ then take $\pi$ to be the identity. If $w'$ is obtained
  from $w$ using one of the transformations, that is, if
  $w\sqsubset w'$, then, for each transformation it is easy to construct a
  function~$\pi$ satisfying the stated properties. If
  $w\sqsubsetneq w'$, then there is a sequence
  $(w_i)_{0\leq i \leq n}$ of observations, with $n\geq 1$, such that
  $w=w_0\sqsubset w_1 \sqsubset \dots \sqsubset w_n=w'$. From the
  previous observation there is a sequence of functions
  $\pi_i:\position(w_{i})\to\position(w_{i-1})$, with $1\leq i\leq n$, each
  satisfying the stated properties. Then it is easy to see that their
  composition $\pi=\pi_1\circ\dots\circ\pi_{n}$ also satisfies these
  properties.
  \qed
\end{proof}

The proof of Theorem~\ref{thm:monotonicity} is similar to that of
Theorem~\ref{thm:osem_refinement} and is thus omitted. We just note
that the omitted proof uses properties~(R1) to (R3) from
Lemma~\ref{lem:refinement}, which correspond to the ones given in the
proof of Theorem~\ref{thm:osem_refinement}.

\subsection{Proof of Theorem~\ref{thm:decidable}}

We first show that the problem is PSPACE-hard by reducing the
satisfiability problem for quantified Boolean logic (QBL) to it.  For
simplicity, we assume that \MTLdata comprises the temporal past-time
connectives $\once$~(``once'') and $\historically$~(``historically''),
which are the counterparts of $\eventually$ and $\always$. Their
semantics is as expected. A slightly more involved reduction without
these temporal connectives is possible.

Let $\alpha$ be a closed QBL formula over propositions
$p_1,\dots,p_n$.  We define the set $P$ of predicate symbols as
$\{P_1,\dots,P_n\}$, where each predicate symbol has arity $1$.
Moreover, let $R:=\{r\}$ and $D:=\{0,1\}$, and let $w$ be the
observation
$(\{0\},\sigma,\rho_0)(\{1\},\sigma,\rho_1)(\{3\},[\,],[\,])((3,\infty),[\,]\,[\,])$,
with $\sigma(P_i)=\{1\}$, for each $i\in\{1,\dots,n\}$, and
$\rho_i(r)=i$, for $i\in\{0,1\}$.  Finally, we translate the QBL
formula $\alpha$ to an \MTLdata formula $\alpha^*$ as follows.
\begin{equation*}
  \begin{array}{@{}r@{\ }lcr@{\ }lcl@{}}
    p_i^* &:= P_i(x_i) && 
    (\neg \beta)^* &:= \neg \beta^* &&
    (\beta\vee\gamma)^* := \beta^*\vee\gamma^*
    \\
    (\exists p_i.\,\beta)^* &:= \once\eventually_{[0,1]}\freeze{}{x_i}{\beta^*}
    &&
    (\forall p_i.\,\beta)^* &:= \historically\,\always_{[0,1]}\freeze{}{x_i}{\beta^*}
  \end{array}  
\end{equation*}
It is easy to see that $\alpha$ is satisfiable iff
$\eosem{w,0,[\,]}{\alpha^*}=\true$.

We only sketch the problem's membership in PSPACE.  Note that $w$ is
finite. If there is no time point in $w$ with timestamp $\tau$, then
$\eosem{w,\tau,\nu}{\phi}=\unknown$. Suppose that $i\in\position(w)$
is a time point in $w$ with timestamp~$\tau$.
A computation of $\phi$'s truth value at position $i$ can be easily
obtained from the inductive definition of the satisfaction
relation~$\omodels$.  This computation can be done in polynomial space
when traversing the formula structure depth-first. Note that for a
position, a subformula of $\phi$ might be visited multiple times with
possibly different partial valuations.


%% file: app_pseudocode.tex
\section{Additional Algorithmic Details}
\label{app:pseudocode}

We present in this section the procedures~\ls{Init},
\ls{NewTimePoint}, \ls{PropagateDown}, and \ls{PropagateUp}, which are
called from the main procedure. Their pseudo-code is given in the
Listings~\ref{fig:proc_init} to~\ref{fig:proc_up}, respectively.

\begin{listing}[t]
  \centering
  \begin{minipage}[t]{0.15\linewidth}
    \caption{}
    \label{fig:proc_init}
  \end{minipage}
  \quad
  \begin{minipage}[t]{0.8\linewidth}
\begin{lstlisting}
procedure Init($\phi$)
  verdicts := $\emptyset$
  $J$ := $[0,\infty)$
  foreach $\psi\in \sub(\phi)$ # in a bottom-up manner
    case $\psi=p(\bar{x})$: $\Psi^{\psi,J,\emap}$ := $\psi^{J}$
    case $\psi=\neg\alpha$: $\Psi^{\psi,J,\emap}$ := $\neg\alpha^J$
    case $\psi=\alpha\lor\beta$: $\Psi^{\psi,J,\emap}$ := $\alpha^J \lor \beta^J$
    case $\psi=\freeze{}{x}{\alpha}$: $\Psi^{\psi,J,\emap}$ := $\alpha^J$
    case $\psi=\alpha\until_I\beta$: 
      if $I = \tl$ then tc := $\true$ else tc := $\tcp^{J,J}$
      $\Psi^{\psi,J,\emap}$ := tc $\land$ $\bar\beta^J \land \bar{\bar\alpha}^J$
  foreach $\alpha$ in Atoms($\phi$) with $\alpha=\true$ do
    PropagateUp($\alpha$, $J$, $\true$)
\end{lstlisting}
  \end{minipage}
\end{listing}

The procedure \ls{Init} initializes the state of the monitor, which
consists of the formulas $\fPsi{}{\gamma,\co{0,\infty},\emap}$, for
each subformula $\gamma$ of the monitored formula~$\phi$.
These formulas are as defined in Section~\ref{subsubsec:state}.
The procedure also propagates the Boolean value~$\true$ from the
atoms~$\true$ of $\phi$.

\begin{listing}[t]
  \centering
  \begin{minipage}[t]{0.15\linewidth}
    \caption{}
    \label{fig:proc_newtp}
  \end{minipage}
  \quad
  \begin{minipage}[t]{0.8\linewidth}
\begin{lstlisting}
procedure NewTimePoint($\phi$, $J$, new)
  foreach $\psi\in \sub(\phi)$ # in a top-down manner
    foreach $\nu$ with $\fPsi{}{\psi,J,\nu}$ $\neq$ nil do
      foreach $K$ in new do
        case $\psi = p(\bar{x})$:
          $\fPsi{}{\psi,K,\nu}$ := Apply($\fPsi{}{\psi,J,\nu}$, $[\psi^J\mapsto \psi^K]$)
        case $\psi = \neg\alpha$:
          $\fPsi{}{\psi,K,\nu}$ := Apply($\fPsi{}{\psi,J,\nu}$, $[\alpha^J\mapsto \alpha^K]$)
        case $\psi = \alpha\lor\beta$:
          $\fPsi{}{\psi,K,\nu}$ := Apply($\fPsi{}{\psi,J,\nu}$, $[\alpha^J\mapsto \alpha^K, \beta^J\mapsto \beta^K]$)
        case $\psi = \freeze{}{x}\alpha$:
          $\fPsi{}{\psi,K,\nu}$ := Apply($\fPsi{}{\psi,J,\nu}$, $[\alpha^J\mapsto \alpha^K]$)
        case $\psi = \alpha\until_I\beta$:
          $\fPsi{}{\psi,K,\nu}$ := Apply($\fPsi{}{\psi,J,\nu}$, $[\tcp^{H,J}\mapsto\tcp^{H,K}]_{\tcp^{H,J}\in \AP(\fPsi{}{\psi,J,\nu})}$)
      foreach $\gamma, H, \mu$ with $\psi^J\in \AP(\fPsi{}{\gamma,H,\mu})$ and ($\gamma=\psi\until_I\_$ or $\gamma=\psi\until_I\_$)
        $\theta$ := RefinementUntil($\gamma$, $H$, $J$, new)
        $\fPsi{}{\gamma,H,\mu}$ := Apply($\fPsi{}{\gamma,H,\mu}$, $\theta$)

procedure RefinementUntil($\alpha\until_I\beta$, $H$, $J$, new)
  anchor, continuation := $\false$, $\true$
  for $K$ in new with $K\geq H$ do
    if Singleton(K) then cont := $\true$ else $\bar{\bar\alpha}^K$
    anchor := anchor $\lor$ $\bar\beta^K \land \tcp^{K,H} \land \mathsf{cont} \land \mathsf{continuation}$
    continuation := continuation $\land$ $\bar\alpha^K$
  return $[\tcp^{J,K}\mapsto\true, \bar\beta^J\mapsto \mathsf{anchor}, \bar\alpha^J\mapsto \mathsf{continuation}, \bar{\bar\alpha}^J\mapsto \true]$
\end{lstlisting}
\end{minipage}
\end{listing}

The procedure \ls{NewTimePoint} transforms formulas
$\fPsi{}{\gamma,K,\nu}$ so that they reflect the interval structure of
the new observation, the one obtained after receiving the current
message. 
Recall that, in the pseudo-code, $J$ is the interval that is split at
the current iteration and the sequence \ls{new} consists of those
intervals among $J\cap\co{0,\tau}$, $\set{\tau}$, and
$J\cap(\tau,\infty)$ that are not complete.
\ls{NewTimePoint} creates a new formula $\fPsi{}{\gamma,K,\nu}$ for
each $K\in\mathsf{new}$ and each each $\nu$ such that the variable
$\fPsi{}{\gamma,J,\nu}$ is defined. The new formula is obtained by
applying a substitution which translates atomic propositions
$\alpha^J$ into propositional formulas over atomic
propositions~$\alpha^{K'}$ with $K'\in\mathsf{new}$. For non-temporal
formulas this propositional formula is simply $\alpha^K$. That is the
substitution is $[\alpha^J\mapsto\alpha^K]$. For temporal formulas the
substitution is obtained in two steps: a first substitution replaces
atomic propositions of the form $\tcp^{H,J}$ with atomic substitutions
of the form~$\tcp^{H,K}$, and a second substitution, obtained by
calling the procedure~\ls{RefinementUntil}, deals with atomic
propositions of the form $\alpha^J$.

\ls{NewTimePoint} also transforms some formulas
$\fPsi{}{\gamma,K,\nu}$ with $K\not\in\mathsf{new}$, that is for
intervals $K$ that occur in the letters of the old observation, that
is, the one from the previous iteration. It does that for those
formulas that have atomic propositions of the form $\alpha^J$. Note
that then it is necessarily the case that $\gamma$ is a temporal
formula. The required substitution is also computed by the
\ls{RefinementUntil} procedure.

The \ls{RefinementUntil} procedure produces the necessary substitution
to update the atomic propositions $\bar{\alpha}^J$,
$\bar{\bar{\alpha}}^J$, and $\bar\beta^J$ that may occur in formulas
$\fPsi{}{\gamma,H,\nu}$ for some interval~$H$ of the new observation,
where $\gamma=\alpha\until_I\beta$. (Note that before calling
\ls{RefinementUntil} the new formulas $\fPsi{}{\gamma,K,\nu}$ with
$K\in\mathsf{new}$ have already been created.)
The substitution of the proposition $\bar\alpha^J$ is
straightforward. Namely, we replace $\bar\alpha^J$ by the conjunction
$\bigwedge_{K\in\mathsf{new},K\geq H} \bar{\alpha}^K$.
Next, the atomic propositions $\tcp^{J,K}$ and $\bar{\bar\alpha}^J$
are discarded: they appeared in $\fPsi{}{\gamma,H,\nu}$ as
conjuncts and thus replacing them by $\true$ effectively discards
them.
The substitution of $\bar\beta^J$ is more involved. We illustrate it
with an example. Suppose that $\mathsf{new} = (K_1,K_2,K_3)$ and that
$H\leq K_1$. Note that $K_2$ is a singleton.
In this case, $\bar\beta^J$ is substituted by the following formula.
\begin{align*}
& \bar{\beta}^{K_1}\land\tcp^{K_1,H}\land\bar{\bar\alpha}^{K_1}
\ \lor \\
& \bar{\beta}^{K_2}\land\tcp^{K_2,H}\land\bar\alpha^{K_1}
\ \lor \\
& \bar{\beta}^{K_3}\land\tcp^{K_3,H}\land\bar{\bar\alpha}^{K_3}\land\bar\alpha^{K_2}\land\bar\alpha^{K_1}
\end{align*}

The application of the computed substitution is performed by the
procedure~\ls{Apply} given in Listing~\ref{fig:proc_up}. \ls{Apply}
actually does more than just applying the substitution given as an
argument. First, it also simplifies the resulting formula. The actual
application and the simplification are performed by procedure
\ls{Substitute}, which is as expected and thus not detailed further.
Second, \ls{Apply} checks whether the resulting formula is a Boolean
constant. If this is the case, then propagation is initiated by
calling the \ls{PropagateUp} procedure.

Note that if $\fPsi{}{\gamma,J,\nu}$ is already a Boolean
constant that has not yet been propagated, because for instance the
$\gamma=\bar\beta$ and the constant is $\true$, then, since \ls{Apply}
tries again to propagate it, this Boolean value will actually be
propagated for the new interval $\set{\tau}$, as it is a singleton and
thus corresponds to a time point.

\begin{listing}[t]
  \centering
  \begin{minipage}[t]{0.15\linewidth}
    \caption{}
    \label{fig:proc_down}
  \end{minipage}
  \quad
  \begin{minipage}[t]{0.8\linewidth}
    \begin{lstlisting}
 procedure PropagateDown($\psi$, $J$, $x$, $d$)
  foreach $\nu$ with $\fPsi{}{\psi, J, \nu}$ $\neq$ nil
    $\fPsi{}{\psi, J, \nu[x\mapsto d]}$ := $\fPsi{}{\psi, J, \nu}$ 
    if $\psi$ $\not\in$ Atoms($\phi$) then
      foreach $\alpha^K \in \AP(\fPsi{}{\psi, J, \nu})$
        PropagateDown($\alpha$, $K$, $x$, $d$)
\end{lstlisting}
\end{minipage}
\end{listing}

The pseudo-code of the procedures \ls{PropagateDown} and
\ls{PropagateUp} is straightforward in that it implements downward and
upward propagation exactly as described in
Section~\ref{subsubsec:main}.

\begin{listing}[t]
  \centering
  \begin{minipage}[t]{0.15\linewidth}
    \caption{}
    \label{fig:proc_up}
  \end{minipage}
  \quad
  \begin{minipage}[t]{0.8\linewidth}
    \begin{lstlisting}
procedure Apply($\fPsi{}{\psi,J,\nu}$, $\theta$)
  $f$ := Substitute($\fPsi{}{\psi,J,\nu}$, $\theta$)
  if $f\in\Two$ then PropagateUp($\psi$, $J$, $\nu$, $f$)
  return $f$

procedure PropagateUp($\psi$, $J$, $b$)
  $\gamma$ := Parent($\psi$)
  if $\gamma$ = nil then
    if Singleton($J$) then verdicts := verdicts $\cup$ $\set{(\mathsf{Timestamp}(J),b)}$
  else if CanPropagateUp($\psi$, $J$, $b$)
    $\theta$ := [$\psi^{J} \mapsto b$]
    foreach $K, \mu$ with $\psi^{J}\in \AP(\fPsi{}{\gamma,K,\mu})$
      $\fPsi{}{\gamma,K,\mu}$ := Apply($\fPsi{}{\gamma,K,\mu}$, $\theta$)

procedure CanPropagateUp($\psi$, $J$, $b$)
  $\gamma$ := Parent($\psi$)    
  case $\gamma$ = nil or $\gamma \neq \_\until_I\_$: return $\true$
  case $\gamma$ = $\_\until_I\psi$: return (Singleton($J$) or not $b$)
  case $\gamma$ = $\psi\until_I\_$: return (Singleton($J$) or $b$)
\end{lstlisting}
\end{minipage}
\end{listing}

%% file: app_proof.tex
\section{Soundness and Completeness Proof}
\label{app:proof}

In this section, we will consider many substitutions from atomic
propositions to propositional formulas.
Given such a substitution~$\theta$ and a propositional formula~$\psi$,
we denote by $\theta(\psi)$ and $\psi\theta$ the formula obtained by
replacing in $\psi$ the atomic propositions $p$ that occur in both
$\psi$ and in $\pdef(\theta_i)$ by $\theta_i(p)$.

\subsection{Overview}

Let $\bar w$ be valid observation sequence and let $\phi$ be the
monitored formula. 
We assume that $\phi$ is not an atomic formula.
We let $(J^i_j,\sigma^i_j,\rho^i_j)$ be the $j$th letter of $w_i$. We
drop the superscript~$i$ if it is clear from the context.
Also, given an interval $J$ in $w_i$, we denote by $\hat{J}$ the index
$j\in\pos(w_i)$ such that $J_j = J$, assuming that the iteration~$i$
is clear from the context.

We first state a lemma from which correctness follows, and only
later prove the lemma.
We recall that $\fPsi{i}{\psi,J,\nu}$ denotes the the value of the
variable $\fPsi{}{\psi,J,\nu}$ from the pseudo-code at the end of
iteration~$i$.

\begin{lemma}\label{lemma:corol}
  For any $i\in\Nat$, $j\in\position(w_i)$, and $b\in\Two$, we have
  that
  $$
  \theta_i^{\emap}(\fPhi{i}{\phi,J}) = b
  \quad \text{iff} \quad
  \fPsi{i}{\phi,J,\emap} = b,
  $$  
  where $J$ is the interval of the $j$th letter of $w_i$.
\end{lemma}

We show next how observational soundness and completeness follow from
this lemma.
We first note that for any iteration $i\in\Nat$, the variable
$\fPsi{}{\psi,J,\nu}$ is defined at the end of iteration~$i$ for
some\footnote{We will later, in Lemma~\ref{lemma:key}, also
  characterize for which $\nu$ is $\fPsi{i}{\psi,J,\nu}$ defined.} 
$\nu$ if and only if $\psi$ is
a subformula of~$\phi$ and $J$ is an interval in~$w_i$.
Next, we note that the global variable~\ls{verdicts} is only updated
in the execution of \ls{PropagateUp($\psi$, $J$, $b$)} when
$\psi=\phi$ and $J$ is a singleton. Furthermore, by analyzing where it
is called from, we see that all calls are preceded by setting
$\fPsi{}{\psi,J,\nu}$ to~$b$.
Moreover, $\nu$ has to be $\emap$ because for any defined variable
$\fPsi{}{\psi,J,\nu}$ with $\nu\neq\emap$ we have that
$\psi\neq\phi$. Indeed, $\fPsi{}{\psi,J,\nu}$ is only used with a
``new'' $\nu$ in \ls{PropagateDown} and this procedure is never called
for $\psi=\phi$.
We thus obtain that if some tuple $(\tau,b)$ is added to
\ls{verdicts} then $\fPsi{}{\phi,\set{\tau},\emap}=b$.

Observational soundness follows directly from the lemma, the above
stated properties of $\theta_i^\nu$ and $\Phi_i$, and the observations
made in the previous paragraph.
Note that, since $\phi$ is closed, then $\osem{w_i,j,\nu}{\phi} =
\osem{w_i,j,\emap}{\phi}$, for any partial valuation~$\nu$.

Consider now completeness. Say that $\osem{w_i,j,\nu}{\phi}=b\in\Two$
and $j$ is a time point with timestamp $\tau$. Let $J=\set{\tau}$.
Then $\osem{w_i,j,\emap}{\phi} = b$ and thus
$\theta_i^{\emap}(\fPhi{i}{\phi,J}) = b$. By the lemma, we have that
$\fPsi{i}{\phi,J,\emap} = b$. Then there is an iteration $i'\le i$
when $\fPsi{i'}{\phi,J',\emap}$ has been set to~$b$, where $J'$ is the
interval in $w_{i'}$ from which $J$ originates.
Let $i''$ be the first iteration when $J$ is a letter of $w_{k}$ for
some $k$. At this iteration $\fPsi{i''}{\phi,J,\emap}$ is set to
$\fPsi{i'}{\phi,J',\emap}$ (see the \ls{NewTimePoint} procedure). Clearly $i'\leq i''\leq i$.
The setting of $\fPsi{i''}{\psi,J,\emap}$ to a new value is preceded
by a call to \ls{Apply}, and since this new value is a Boolean value,
\ls{Apply} calls \ls{PropagateUp} which adds $(\tau,b)$ to
\ls{verdicts}.

\subsection{The Key Lemma}

We state next the key lemma of the proof, from which
Lemma~\ref{lemma:corol} follows.
In order to state this more general lemma, we first introduce some
additional notation.

\subsubsection{Additional Notation.}

Given an $i\in\Nat$, a subformula $\psi$ of $\phi$, and a
$j\in\position(w_i)$, we define inductively the set $\Val_i(\psi,j)$
of \emph{relevant valuations} for $\psi$ at iteration~$i$ and
position~$j$.
For $i=0$, we set $\Val_i(\psi,0) := \set{\emap}$, for any subformula
$\psi$ of $\phi$, and for $i>0$ we let
\begin{align*}
  \Val'_i(\psi,j) & := 
  \left\{
  \begin{array}{l@{\quad}l}
    \set{[\,]} & \text{if $\psi=\phi$},
    \\
    \Val_i(\gamma,j) & \text{if $\gamma=\neg\psi$, $\gamma=\psi\lor\psi'$, or $\gamma=\psi'\lor\psi$},
    \\
    \setx{\nu[x\mapsto \rho^i_j(r)]}{\nu\in\Val_i(\gamma,j)} & \text{if $\gamma=\freeze{x}{r}{\psi}$}
    \\
    \cup_{h\in A}\Val_i(\gamma,h) & \text{if $\gamma=\psi\until_I\psi'$}
    \\
    \cup_{k\in B}\Val_i(\gamma,k) & \text{if $\gamma=\psi'\until_I\psi$}
  \end{array}
  \right.
  \\
  \text{and} \qquad\qquad &
  \\
  \Val_i(\psi,j) & := \setx{\nu\in\Val'_i(\psi,j)}{\osem{w_i,j,\nu}{\psi}\not\in\Two},
\end{align*}
where 
$\gamma$ is the ``parent'' of $\psi$ (i.e.~the subformula of $\phi$
that has $\psi$ as a direct subformula) and
\begin{align*}
  A := &\; \setx{h\in\Nat}{k<h\leq j \text{ where $k\in\Nat$ is such that $k\leq j$ and $\tc_{w_i,I}(j,k)\neq\false$}}, \\
  B := &\; \setx{k\in\Nat}{k\leq j \text{ and $\tc_{w_i,I}(j,k)\neq\false$}}.
\end{align*}

We denote by $\AP_i$ the set of the atomic propositions used by the
formulas~$\fPhi{i}{\gamma,J}$.
We let $\bAP_i$ be the set of atomic propositions obtained from
$\AP_i$ be removing the atomic propositions of the form $\tpp^J$ and
replacing atomic propositions of the form $\beta^J$ and $\alpha^J$
where $\alpha\until_I\beta$ is an subformula of $\phi$ for some $I$,
by the atomic propositions $\bar\beta^J$ and respectively
$\bar\alpha^{J}$ and $\bar{\bar\alpha}^{J}$.
Note that $\bAP_i$ represents the set of atomic propositions of the
formulas $\fPsi{i}{\gamma,J,\nu}$ from the pseudo-code.

We also let $\deltap_i^\nu$ be the following substitution: 
\begin{align*}
\deltap_i^\nu(\bar\beta^J) := &
\left\{
\begin{array}{l@{\quad}l}
\tpp^J\land\beta^J & \text{if $J$ not a singleton and $\fPsi{i}{\beta,J,\nu}\notin\Two$}
\\
\tpp^J & \text{if $J$ not a singleton and $\fPsi{i}{\beta,J,\nu}=\true$}
\\
\beta^J & \text{if $J$ is a singleton and $\fPsi{i}{\beta,J,\nu}\notin\Two$}
\\
\text{undefined} & \text{otherwise}
\end{array}
\right.
\\
\deltap_i^\nu(\bar\alpha^J) := &
\left\{
\begin{array}{l@{\quad}l}
\tpp^J\to\alpha^J & \text{if $J$ not a singleton and $\fPsi{i}{\alpha,J,\nu}\notin\Two$}
\\
\neg\tpp^J & \text{if $J$ not a singleton and $\fPsi{i}{\alpha,J,\nu}=\false$}
\\
\alpha^J & \text{if $J$ is a singleton and $\fPsi{i}{\alpha,J,\nu}\notin\Two$}
\\
\text{undefined} & \text{otherwise}
\end{array}
\right.
\\
\deltap_i^\nu(\bar{\bar\alpha}^J) := &
\left\{
\begin{array}{l@{\quad}l}
\tpp^J\lor\alpha^J & \text{if $J$ not a singleton and $\fPsi{i}{\alpha,J,\nu}\notin\Two$}
\\
\tpp^J & \text{if $J$ not a singleton and $\fPsi{i}{\alpha,J,\nu}=\false$}
\\
\alpha^J & \text{if $J$ is a singleton and $\fPsi{i}{\alpha,J,\nu}\notin\Two$}
\\
\text{undefined} & \text{otherwise}
\end{array}
\right.
\end{align*}
Note that, for instance when $J$ is a singleton and
$\fPsi{i}{\alpha,J,\nu}=\true$, there is no need to define
$\deltap_i^\nu(\bar\alpha^J)$, because in this case $\bar\alpha^J$ is not
an atomic proposition of $\fPsi{i}{\alpha\until_I\beta,K,\nu}$: that
atomic proposition has already been instantiated.

\subsubsection{The Lemma and its Proof.}

The following key lemma states the main invariant of the monitoring algorithm.

\begin{lemma}\label{lemma:key}
  For any non-atomic subformula $\psi$ of $\phi$, $i\in\Nat$,
  $j\in\position(w_i)$, partial valuation~$\nu\in\Val_i(\psi,j)$, we
  have that
  $$
  \theta_i^\mu(\fPhi{i}{\psi,J})
  \equiv
  \deltap_i^\mu(\fPsi{i}{\psi,J,\nu}),
  $$
  where $J$ is the interval of the $j$th letter of $w_i$ and $\mu =
  \nu[x\mapsto \rho^i_j(r)]$ if $\psi=\freeze{r}{x}{\alpha}$ and
  $\mu=\nu$ otherwise.
\end{lemma} 

We first note that Lemma~\ref{lemma:corol} follows easily from
Lemma~\ref{lemma:key}. Indeed, we just need to show that,
$\deltap_i^\mu(\fPsi{i}{\phi,J,\emap}) = b$ iff
$\fPsi{i}{\phi,J,\emap}=b$, for any $b\in\Two$ where $\mu$ is an in
the lemma statement. The right to left direction is trivial, while the
left to right direction also follows easily: if
$\fPsi{i}{\phi,J,\emap}\notin\Two$ then it contains at least an atomic
proposition, and then $\deltap_i^{\emap}(\fPsi{i}{\phi,J,\emap})$ also
contains an atomic proposition, by the definition of
$\deltap_i^{\mu}$, and thus
$\deltap_i^\mu(\fPsi{i}{\phi,J,\emap})\notin\Two$ too.

We devote the rest of this section to the proof of Lemma~\ref{lemma:key}.
We start with an assumption and continue with a series of definitions.

We assume in the pseudo-code that formula rewriting (that is, applying
substitutions) and propagations are separated into two distinct
phases. Note that this can be easily achieved by postponing
propagations, that is, by storing all tuples $(\phi, J, f)$ with
$f\in\Two$ that occur during the execution of \ls{Apply}, and by
performing the corresponding propagations at the end of the
\ls{NewTimePoint} procedure. We choose not to present this two-phase
version of the pseudo-code in order to keep the pseudo-code more
compact.

Let $\mathsf{new}_i$ be the value of \ls{new} at iteration~$i>0$.
Given an interval $J\in\mathsf{new}_i$, we let $\zeta^J_i$ be the
substitution computed by the pseudo-code at some iteration~$i>0$
during the formula rewriting phase. Note that there is no formula
rewriting phase at iteration $i=0$. For $i>0$, for instance when
$\alpha$ is such that its parent is $\psi=\neg\alpha$, we have that
$\zeta^J_i(\alpha^{J'})=\alpha^J$ when $J$ is a new interval in $w_i$
obtained by splitting the interval $J'$ in $w_{i-1}$.\footnote{Note
  that in the case of temporal subformulas, this is the composition of
  two substitutions: the unnamed one and the substitution~$\theta$
  computed by \lsf{RefinementUntil}.}

We let $\bar\xi_i^\nu$ be the substitution computed by the pseudo-code
at iteration~$i$ during the propagation phase for the partial
valuation~$\nu$.
Formally, $\bar\xi_i^\nu$ is a partial valuation over $\bAP_i$ defined
as follows:
\begin{align*}
    \xi_i^\nu(p(\bar{x})^J) & := \hspace*{0.8em} \true\qquad \text{if $\bar{x}\in \pdef(\phi)$, $p\in\pdef(\sigma)$, and $\nu(\bar{x})\in \sigma(p)$},
  \\
  \xi_i^\nu(p(\bar{x})^J) & := \hspace*{0.8em} \false\qquad \text{if $\bar{x}\in \pdef(\phi)$, $p\in\pdef(\sigma)$, and $\nu(\bar{x})\not\in \sigma(p)$},
  \\
  \bar\xi_i^\nu(\gamma^J) & := \hspace*{0.8em} b\qquad \text{if $\gamma$ is non-atomic and $\fPsi{i}{\gamma,J,\nu}=b\in\Two$},
  \\
  \bar\xi_i^\nu(\bar\beta^J) & :=
  \left\{
  \begin{array}{l@{\qquad}l}
    \true  & \text{if $J$ is a singleton and $\fPsi{i}{\beta,J,\nu}=\true$},
    \\
    \false & \text{if $\fPsi{i}{\beta,J,\nu}=\false$},
  \end{array}\right.
  \\
  \bar\xi_i^\nu(\bar\alpha^J) :=
  \bar\xi_i^\nu(\bar{\bar\alpha}^J) & :=
  \left\{
  \begin{array}{l@{\qquad}l}
    \true  & \text{if $\fPsi{i}{\alpha,J,\nu}=\true$},
    \\
    \false & \text{if $J$ is a singleton and $\fPsi{i}{\alpha,J,\nu}=\false$},
  \end{array}\right.
  \\
  \bar\xi_i^\nu(\tcp_{\alpha\until_I\beta}^{K,J}) & := 
  \left\{
    \begin{array}{l@{\qquad}l}
      \true & \text{if $(K-J)\subseteq I$},
      \\
      \false & \text{if $(K-J)\cap I = \emptyset$}.
    \end{array}
  \right.
\end{align*}

For any $i>0$, subformula~$\psi$ of $\phi$, and interval~$J$ in $w_i$,
we have
\begin{equation}
  \label{eq:Psi}
  \fPsi{i}{\psi,J,\nu} := \bar\xi_i^\mu(\zeta^J_i(\fPsi{i-1}{\psi,J',\nu})),
\end{equation}
where $J'$ is the interval in $w_{i-1}$ from which $J$ originates and
$\mu = \nu[x\mapsto \rho^i_j(r)]$ if $\psi=\freeze{r}{x}{\alpha}$ and
$\mu=\nu$ otherwise.

We define next the substitution~$\xi_i^\nu$ over $\AP_i$ and which is
a variant of $\bar\xi_i^\nu$:
\begin{align*}
  \xi_i^\nu(\gamma^J) & := \bar\xi_i^\nu(\gamma^J) \quad  \text{if $\gamma\in\sub(\phi)$},
  \\
  \xi_i^\nu(\tpp^{J}) & := \true \hspace*{3.5em} \text{if $J$ is a singleton},
  \\
  \xi_i^\nu(\tcp_{\psi}^{K,J}) & := \bar\xi_i^\nu(\tcp_{\psi}^{K,J}). &
\end{align*}

We also let $\bar\delta_i$ be the following substitution: 
\begin{align*}
\bar\delta_i(\bar\beta^J) & := \tpp^J\land\beta^J,
\\ 
\bar\delta_i(\bar\alpha^J) & := \tpp^J\to\alpha^J,
\\
\bar\delta_i(\bar{\bar\alpha}^J) & := \tpp^J\lor\alpha^J.
\end{align*}
Note that we do not have that $\deltap_i^\nu = \xi_i^\nu \circ
\bar\delta_i$. Indeed, if $\Psi_i^{\beta,J,\nu}=\false$ then
$\bar\beta^J\, \deltap_i^\nu = \bar\beta^J$, while $\bar\beta^J\,
\bar\delta_i\, \xi_i^\nu = \false$.
However, as we will prove later, we have that $\deltap_i^\nu \circ
\bar\xi_i^\nu = \xi_i^\nu \circ \bar\delta_i$.

For $i>0$ and partial valuation~$\nu$, we define the
substitution~$\theta_{i-1,i}^\nu$ such that it behaves
as~$\theta_{i-1}^\nu$, but it acts however on the atomic
propositions of formulas~$\fPhi{i}{\gamma,J,\nu}$. In other words, in
contrast to~$\theta_i^\nu$, it does not take into account the new
interpretations received at iteration~$i$. Formally,
$\theta_{i-1,i}^\nu$ is a partial function over~$\AP_i$ defined by
\begin{align*}
  \theta_{i-1,i}^\nu(\alpha^{J}) := & \  \theta_{i-1}^\nu(\alpha^{J'})
  \\
  \theta_{i-1,i}^\nu(\tpp^{J}) := & \  \theta_{i-1}^\nu(\tpp^{J'})
  \\
  \theta_{i-1,i}^\nu(\tcp^{H,K}) := & \  \theta_{i-1}^\nu(\tcp^{H',K'}), 
\end{align*}
where $J'$, $H'$, and $K'$ are the intervals in $w_{i-1}$ from which
the intervals $J$, $H$, and respectively $K$ originate.

Finally, we let $\theta_{i,\Delta}$ be the substitution
defined over $\AP_i$ by
\[\begin{array}{rl@{\quad}l}
\theta_{i,\Delta}^\nu(\alpha^{J}) := & \theta_i^\nu(\alpha^{J}) & \text{if $i=0$ or $\theta_{i-1}^\nu$ is undefined on $\alpha^{J'}$,}
\\
\theta_{i,\Delta}^\nu(\tpp^J) := & \theta_i^\nu(\tpp^J) & \text{if $i=0$ or $\theta_{i-1}^\nu$ is undefined on $\tpp^{J'}$,}
\\
\theta_{i,\Delta}^\nu(\tcp^{H,K}) := & \theta_i^\nu(\tcp^{H,K}) & \text{if $i=0$ or $\theta_{i-1}^\nu$ is undefined on $\tcp^{H',K'}$,}
\end{array}\]
where, when $i>0$, $J'$, $H'$, and $K'$ are the intervals in $w_{i-1}$ from which
the intervals $J$, $H$, and respectively $K$ originate.

Clearly, we have that
\begin{align}
\label{eq:theta}
\theta_i^\nu & = \theta_{i-1,i}^\nu\, \theta_{i,\Delta}^\nu,\quad \text{for $i>0$},
\\
\label{eq:theta0}
\theta_0^\nu & = \theta_{0,\Delta}^\nu.
\end{align}

The following easy to prove equivalence will also be useful:
\begin{equation}
  \label{eq:Phi}
  \fPhi{i}{\psi,J}\theta_i^\mu = b
  \quad\text{iff}\quad
  \psi^J \theta_i^\nu = b,
\end{equation}
for any formula $\psi$ that is a subformula of $\phi$ and any $b\in\Two$, and
where $\mu = \nu[x\mapsto \rho^i_j(r)]$ if
$\psi=\freeze{r}{x}{\alpha}$ and $\mu=\nu$ otherwise.

We summarize in Table~\ref{tab:substs} all
substitutions used in the proof together with their informal meaning.
\begin{table}[t]
  \caption{Summary of substitutions used in the proof.}
  \label{tab:substs}
  \centering
  \begin{tabular}{|l@{\ }|@{\ }l@{\ }|@{\ }l|}
    \hline
    notation & domain\ & intuition  
    \\ \hline\hline \rule{0pt}{2.5ex}
    $\theta_i^\nu$ & $\AP_i$ & $\gamma^J \mapsto \osem{w_i,\hat{J},\nu}{\gamma}$ if in $\Two$
    \\ \hline \rule{0pt}{2.5ex}
    $\theta_{i-1,i}^\nu$ & $\AP_i$ & ''extension'' of $\theta_{i-1}^\nu$ from $A_{i-1}$ to $\AP_i$
    \\ \hline \rule{0pt}{2.5ex}
    $\theta_{i,\Delta}^\nu$ & $\AP_i$ & a substitution such that $\theta_i^\nu = \theta_{i-1,i}^\nu\, \theta_{i,\Delta}^\nu$
    \\ \hline \rule{0pt}{2.5ex}
    $\bar\zeta_i^{J}$ & $\bAP_{i-1}$ & the substitution of the rewriting phase 
    \\ \hline \rule{0pt}{2.5ex}
    $\zeta_i^J$ & $\AP_{i-1}$ & variant of $\bar\zeta_i^J$ over $\AP_{i-1}$ such that $\bar\zeta^J_i\, \bar\delta_i \equiv \bar\delta_{i-1}\, \zeta^{J}_{i}$ 
    \\ \hline \rule{0pt}{2.5ex}
    $\bar\xi_i^\nu$ & $\bAP_i$ & the substitution of the propagation phase
    \\ \hline \rule{0pt}{2.5ex}
    $\xi_i^\nu$ & $\AP_i$ & $\gamma^J \mapsto \fPsi{i}{\gamma,J,\nu}$ if in $\Two$
    \\ \hline \rule{0pt}{2.5ex}
    $\bar\delta_i$ & $\bAP_i$ & e.g. $\bar\beta^J\mapsto \tpp^J\land\beta^J$
    \\ \hline \rule{0pt}{2.5ex}
    $\deltap_i^\nu$ & $\bAP_i$ & variant of $\bar\delta_i$, taking propagations into account
    \\ \hline
  \end{tabular}
\end{table}

In the remaining proof, we will use equivalences between
substitutions, with the following meaning.
Given two substitutions $\theta$ and $\theta'$ from atomic
propositions to propositional formulas, we write that
$\theta\equiv\theta'$ iff $\pdef(\theta)=\pdef(\theta')$ and $p\theta
\equiv p\theta'$, for any $p\in\pdef(\theta)$.

\medskip

Next, we prove Lemma~\ref{lemma:key} by nested induction, the outer
induction being on the iteration~$i$ and the inner induction being a
structural induction on $\psi$.

\emph{Outer base case: $i=0$}. We have a single letter in $w_0$ with the
interval $J=[0,\infty)$. 
Let $\fPsip{0}{\psi,J,\emap}$ be the value of $\fPsi{}{\psi,J,\emap}$ at the
point of execution of the \ls{Init} procedure (thus at iteration~$0$)
between the two \ls{foreach} loops, that is, before propagation of the
$\true$~atoms.
We thus have that 
\begin{equation}\label{eq:Psi0}
\fPsi{0}{\gamma,J,\emap} = \fPsip{0}{\gamma,J,\emap}\, \bar\xi_0^{\emap}.
\end{equation}
We also have that
\begin{equation}\label{eq:Psi0-Phi0}
\fPsip{0}{\gamma,J,\emap}\bar\delta_0 \equiv \fPhi{0}{\gamma,J}.
\end{equation}
This is easy to check by inspecting the \ls{Init} procedure. For
instance, for formulas $\gamma = \alpha\until_I\beta$, we have that
$\fPsip{0}{\gamma,J} = \tcp^{J,J} \land \bar\beta^{J} \land
\bar{\bar\alpha}^{J} $ and thus
$$
\fPsip{0}{\gamma,J}\bar\delta_0 
= 
\tcp^{J,J} \land (\tpp^J \land \beta^{J}) \land (\tpp^J\lor \alpha^{J}) 
\equiv 
\tpp^J \land \tcp^{J,J} \land \beta^{J} 
= 
\fPhi{0}{\gamma,J}.
$$

This case then follows from the following sequence of equivalences:
\[
\begin{array}{r@{\quad}l}
  \blue{\fPsi{0}{\gamma,J,\emap}}\, \deltap_0^{\emap}
  \equiv & \byref{\ref{eq:Psi0}}{\fPsi{0}{\gamma,J,\emap} = \fPsip{0}{\gamma,J,\emap}\, \bar\xi_0}
  \\[0.5ex]
  \fPsip{0}{\gamma,J}\, \blue{\bar\xi_0^{\emap}\, \deltap_0}
  \equiv & \byref{\ref{eq:xi-delta}}{\bar\xi_i^\nu\, \deltap_i^\nu \equiv \bar\delta_i\, \xi_i^\nu}
  \\[0.5ex]
  \blue{\fPsip{0}{\gamma,J,\emap}\, \bar\delta_0}\, \xi_0^{\emap}
  \equiv & \byref{\ref{eq:Psi0-Phi0}}{\fPsip{0}{\gamma,J,\emap}\, \bar\delta_0 \equiv \fPhi{0}{\gamma,J}}
  \\[0.5ex]
  \fPhi{0}{\gamma,J}\, \blue{\xi_{0}^{\emap}}
  \equiv & \byvar{(\ref{eq:theta0}) and (\ref{eq:IHp})}{\xi_{0}^{\emap} \equiv \theta_0^{\emap}}
  \\[0.5ex]
  \fPhi{i}{\gamma,J}\, \theta_{0}^{\emap}
  \ \phantom{\equiv}
\end{array}
\]
We postpone the proof of the not yet justified equivalences (namely
the 2nd and 4th), as similar ones are also used in the inductive case,
and are proved together.

\emph{Outer inductive case: $i>0$}. We assume that the 
equivalence from the lemma statement holds for $i-1$:
\begin{equation}
\label{eq:IH}
\tag{IH}
\fPsi{i-1}{\gamma,J',\nu}\, \deltap_{i-1}^\nu \equiv \fPhi{i-1}{\gamma,J'}\, \theta_{i-1}^\nu,
\end{equation}
where $J'$ is the intervals in~$w_{i-1}$ from which $J$ originates.
The inductive case follows from the following sequence of equivalences:
\[
\begin{array}{r@{\quad}l}
  \blue{\fPsi{i}{\gamma,J,\nu}}\, \deltap_i^\mu
  \equiv 
  & \byref{\ref{eq:Psi}}{\fPsi{i}{\gamma,J,\nu} = \fPsi{i-1}{\gamma,J',\nu}\, \bar\zeta^J_i\, \bar\xi_i^\mu}
  \\[0.5ex]
  \fPsi{i-1}{\gamma,J',\nu}\, \bar\zeta^J_i\, \blue{\bar\xi_i^\mu\, \deltap_i^\mu}
  \equiv 
  & \byref{\ref{eq:xi-delta}}{\bar\xi_i^\mu\, \deltap_i^\mu \equiv \bar\delta_i\, \xi_i^\mu}
  \\[0.5ex]
  \fPsi{i-1}{\gamma,J',\nu}\, \bar\zeta^J_i\, \bar\delta_i\, \blue{\xi_i^\mu}
  \equiv 
  & \byref{\ref{eq:xi-idem}}{\xi_i^\mu \equiv \xi_i^\mu\, \xi_i^\mu}
  \\[0.5ex]
  \fPsi{i-1}{\gamma,J',\nu}\, \blue{\bar\zeta^J_i\, \bar\delta_i}\, \xi_i^\mu\, \xi_i^\mu
  \equiv 
  & \byref{\ref{eq:zeta-delta}}{\bar\zeta^J_i\, \bar\delta_i \equiv \bar\delta_{i-1}\, \zeta^{J}_{i}}
  \\[0.5ex]
  \fPsi{i-1}{\gamma,J',\nu}\, \bar\delta_{i-1}\, \blue{\zeta^{J}_{i}\, \xi_i^\mu} \, \xi_i^\mu
  \equiv 
  & \byref{\ref{eq:xi-zeta}}{\zeta^{J}_{i}\, \xi_i^\mu \equiv \xi_{i-1}\, \zeta^{J}_{i}}
  \\[0.5ex]
  \blue{\fPsi{i-1}{\gamma,J',\nu}\, \bar\delta_{i-1}\, \xi_{i-1}}\, \zeta^{J}_{i}\, \xi_i^\mu
  \equiv 
  & \byref{\ref{eq:Psi-xi}}{\fPsi{i-1}{\gamma,J',\nu}\bar\delta_{i-1}\, \xi_{i-1} \equiv \fPsi{i-1}{\gamma,J',\nu}\deltap_{i-1}}
  \\[0.5ex]
  \blue{\fPsi{i-1}{\gamma,J',\nu}\, \deltap_{i-1}}\, \zeta^{J}_{i}\, \xi_i^\mu
  \equiv 
  & \byref{\ref{eq:IH}}{\fPsi{i-1}{\gamma,J',\nu}\, \deltap_{i-1} \equiv \fPhi{i-1}{\gamma,J'}\, \theta_{i-1}^\mu}
  \\[0.5ex]
  \fPhi{i-1}{\gamma,J'}\, \theta_{i-1}^\mu\, \zeta^{J}_{i}\, \blue{\xi_i^\mu}
  \equiv 
  & \byref{\ref{eq:IHp}}{x\,\xi_i^\mu \equiv x\,\theta_{i,\Delta}^\mu, \text{for any $q\in\AP(\fPhi{i-1}{\gamma,J'}\, \theta_{i-1}^\mu\, \zeta^{J}_{i})$}}
  \\[0.5ex]
  \fPhi{i-1}{\gamma,J'}\, \blue{\theta_{i-1}^\mu\, \zeta^{J}_{i}}\, \theta_{i,\Delta}^\mu
  \equiv 
  & \byref{\ref{eq:zeta-theta}}{\theta_{i-1}^\mu\, \zeta^{J}_{i} \equiv \zeta^{J}_{i}\, \theta_{i-1,i}^\mu}
  \\[0.5ex]
  \blue{\fPhi{i-1}{\gamma,J'}\, \zeta^{J}_{i}}\, \theta_{i-1,i}^\mu\, \theta_{i,\Delta}^\mu
  \equiv 
  & \byref{\ref{eq:Phi-zeta}}{\fPhi{i-1}{\gamma,J'}\, \zeta^{J}_{i} \equiv \fPhi{i}{\gamma,J}}
  \\[0.5ex]
  \fPhi{i}{\gamma,J}\, \blue{\theta_{i-1,i}^\mu\, \theta_{i,\Delta}^\mu}
  \equiv 
  & \byref{\ref{eq:theta}}{\theta_{i-1,i}^\mu\, \theta_{i,\Delta}^\mu \equiv \theta_i^\mu}
  \\[0.5ex]
  \fPhi{i}{\gamma,J}\, \theta_i^\mu
  \ \phantom{\equiv}
\end{array}
\]
where $\mu = \nu[x\mapsto \rho^i_j(r)]$ if
$\psi=\freeze{r}{x}{\alpha}$ and $\mu=\nu$ otherwise, and $\zeta^J_i$
is a substitution depending on $\bar\zeta^J_i$.

We now complete the proof by proving the not yet justified
equivalences occurring in the above two sequences of equivalences.
We start with the following statement. For any $q\in\AP_i$ such that
$q=\tpp^K$ or $q=\tcp^{H,K}$, or $q=\psi^K$ with $\psi$ a direct
subformula of $\gamma$, the following holds:
\begin{equation}
  \label{eq:IHp}
  \tag{IH'}
  q\,\xi_i^\mu \equiv q\,\theta_{i,\Delta}^\mu.
\end{equation}

The case when $q$ is one of the atomic propositions $p(\bar{x})^J$
with $p(\bar{x})$, $\tpp^{J}$, and $\tcp^{H,K}$, follows directly from
the definitions of~$\xi_i^\mu$ and~$\theta_{i,\Delta}^\mu$.
So let $q=\psi^K$. We have that $\psi^K\, \xi_i^\mu$ equals
$\fPsi{i}{\psi,K,\mu}$ if $\fPsi{i}{\psi,K,\mu}\in\Two$ and
equals~$\psi^K$ otherwise.

Suppose first that $\fPsi{i}{\psi,K,\mu} = b$, for some
$b\in\Two$. Then, by the inner induction hypothesis, we have that
$\fPhi{i}{\psi,K}\theta_i^{\eta} \equiv
\fPsi{i}{\psi,K,\mu}\deltap_i^{\eta} = b$, where $\eta =
\mu[x\mapsto \rho^i_j(r)]$ if $\psi=\freeze{r}{x}{\alpha}$ and
$\eta=\nu$ otherwise.
From~(\ref{eq:Phi}), we know that $\fPhi{i}{\psi,K}\theta_i^{\eta}
\equiv \psi^K\theta_i^{\mu}$, and thus we obtain that
$\psi^K\theta_i^\mu = b$.
   
Suppose now that $\fPsi{i}{\psi,K,\mu}\not\in\Two$. Then, 
reasoning similarly to the previous case, we obtain that
$\psi^K\theta_i^\mu\notin\Two$. Thus $\psi^K\theta_i^\mu = \psi^K$ and
hence also $\psi^K\theta_{i,\Delta}^\mu = \psi^K$.

\subsubsection{Remaining Details.}

\makeatletter
\tagsleft@true
\makeatother

The following two equivalences follow directly from the definitions of
the four involved substitutions.
\begin{align}
  \label{eq:xi-idem}
  \xi_i^\nu & \equiv \xi_i^\nu\, \xi_i^\nu
  \\
  \label{eq:xi-delta}
  \bar\xi_i^\nu\, \deltap_i^\nu & \equiv \bar\delta_i\, \xi_i^\nu
\end{align}

Furthermore, for any $i\in\Nat$, $J$ interval in $w_i$, partial
evaluations~$\nu$ and subformula $\gamma$ of $\phi$, the following
equivalence holds:
\begin{align}
  \label{eq:Psi-xi}
  \fPsi{i}{\gamma,J,\nu}\, \deltap_i^\mu & \equiv
                                           \fPsi{i}{\gamma,J,\nu}\bar\delta_i\,
                                           \xi_i^\mu .
\end{align}
Indeed, let $\fPsip{i}{\gamma,J,\nu} :=
\fPsi{i-1}{\gamma,J',\nu}\bar\zeta^{J}$, for $i>0$. From
(\ref{eq:Psi}) and~(\ref{eq:Psi0}), we get that, for any $i\geq 0$, we
need to prove that $\fPsip{i}{\gamma,J,\nu}\, \bar\xi_i^\nu\,
\deltap_i^\nu \equiv \fPsip{i}{\gamma,J,\nu}\, \bar\xi_i^\nu\,
\bar\delta_i\, \xi_i^\nu$.
This follows directly from equivalence~(\ref{eq:xi-delta}) and the
equivalence $\bar\xi_i^\nu\equiv \bar\xi_i^\nu\, \bar\xi_i^\nu$.

\begin{lemma}
  For any~$i>0$ and~$J$ in $\mathsf{new}_i$, there is a
  substitution~$\zeta^{J}_i$ such that, for any formula~$\gamma$ and
  partial valuation~$\nu$, the following equivalences hold:
  \begin{align}
    \label{eq:xi-zeta}
    \zeta^{J}_{i}\, \xi_i^\nu & \equiv \xi_{i-1}^\nu\, \zeta^{J}_{i}
    \\
    \label{eq:zeta-theta}
    \zeta^{J}_{i}\, \theta_{i-1,i}^\nu & \equiv \theta_{i-1}^\nu\, \zeta^{J}_{i}
    \\
    \label{eq:zeta-delta}
    \bar\zeta^J_i\, \bar\delta_i & \equiv \bar\delta_{i-1}\, \zeta^{J}_{i}
    \\
    \label{eq:Phi-zeta}
    \fPhi{i}{\gamma,J} & \equiv \fPhi{i-1}{\gamma,J'}\, \zeta^{J}_{i},
  \end{align}
  where $J'$ is the interval in $w_{i-1}$ from which $J$ originates.
\end{lemma}
\begin{proof}
  Let $\psi$ be a proper subformula of $\phi$ and $\gamma$ its parent.
  For readability, in this proof we drop the index~$i$ from
  $\bar\zeta_i^J$ and $\zeta_i^J$.
  If $\psi$ is a direct subformula of a non-temporal subformula of
  $\phi$, then $\bar\zeta^J(\psi^{J'}) = \psi^J$. In this case we let
  $\zeta^J(\psi^{J'}) := \psi^J$.
  It is easy to check that the four equivalences hold in this case.

  We consider now the case when $\gamma$ is a temporal formula, with
  $\gamma=\alpha\until_I\beta$.
  We first note that for atomic propositions $\bar\beta^{J'}$,
  $\bar\alpha^{J'}$, and $\bar{\bar\alpha}^{J'}$, the
  substitution~$\bar\zeta^J_i$ depends on whether $J$ is $L$, $N$,
  or~$R$, where $(L,N,R)$ are the intervals obtained by splitting~$J$.
  Thus, we make a case distinction based on the value of~$J$. We only
  consider one case, namely when $J=N$, the other ones being treated
  similarly.
  In this case we have the following equalities:
  \begin{align*}
    \bar\zeta^N(\bar\beta^{J'}) & = (\bar\beta^N \land \tcp^{N,N}) \lor (\bar\beta^R \land \tcp^{R,N} \land \bar{\bar\alpha}^R \land \bar\alpha^N),
    \\
    \bar\zeta^N(\bar\alpha^{J'}) & = \bar\alpha^N \land \bar\alpha^R,
    \\
    \bar\zeta^N(\bar{\bar\alpha}^{J'}) & = \true, 
    \\
    \bar\zeta^N(\tcp^{{J'},{J'}}) & = \true,
    \\
    \bar\zeta^N(\tcp^{H,{J'}}) & = \tcp^{H,N},\ \text{for $H>{J'}$},
  \end{align*}
  and the substitution $\zeta^N$ is defined as follows
  \begin{align*}
    \zeta^N(\beta^{J'}) & := 
      (\beta^N \land \tpp^{N} \land \tcp^{N,N})\ \lor 
    \\ & \phantom{:=}                 
      (\beta^R \land \tpp^{R} \land \tcp^{R,N} \land 
      (\tpp^R\lor\alpha^R) \land (\tpp^N\to\alpha^N)),
    \\
    \zeta^N(\alpha^{J'}) & := (\tpp^N\to\alpha^N) \land (\tpp^R\to\alpha^R),
    \\
    \zeta^N(\tpp^{{J'}}) & := \true, 
    \\
    \zeta^N(\tcp^{J',J'}) & := \true, 
    \\
    \bar\zeta^N(\tcp^{H,{J'}}) & := \tcp^{H,N},\ \text{for $H>{J'}$}.
  \end{align*}

  By just using the definitions, it is easy to check that the
  equivalences~(\ref{eq:xi-zeta}), (\ref{eq:zeta-theta}), and
  (\ref{eq:zeta-delta}) hold.
  For instance, we check next that
  $\bar\beta^{J'}\, \bar\zeta^J\, \bar\delta_i \equiv \bar\beta^{J'}\,
  \bar\delta_{i-1}\, \zeta^{J}$: 
  \begin{align*}
    \bar\delta_i(\bar\zeta^N(\bar\beta^{J'})) 
    & = 
      \bar\delta_i\big((\bar\beta^N \land \tcp^{N,N}) \lor
      (\bar\beta^R \land \tcp^{R,N} \land \bar{\bar\alpha}^R \land
      \bar\alpha^N)\big)
    \\ & = 
        ((\tpp^N \land \beta^N) \land \tcp^{N,N})\ \lor
    \\ & \qquad
         ((\tpp^R\land\beta^R) \land \tcp^{R,N} \land (\tpp^R \lor\alpha^R) \land
         (\tpp^N \to\alpha^N))
    \\ & 
        = \zeta^N(\beta^{J'})
        = \zeta^N(\tpp^{J'}\land\beta^{J'})
        = \zeta^N(\bar\delta_{i-1}(\bar\beta^{J'})).
  \end{align*}

  For~(\ref{eq:Phi-zeta}), we check next that
  $\fPhi{i}{\gamma,N} \equiv \fPhi{i-1}{\gamma,J'}\zeta^N$:
  \begin{align*}
    \fPhi{i}{\gamma,N} \equiv\ & 
    \bigvee_{K\geq N}\big(\tpp^{K} \land \tcp^{K,N} \land \beta^{K} \land
    \bigwedge_{N\leq H<K}(\tpp^{H} \to \alpha^{H})
    \big)
    \\ \equiv\ & 
    \big(\tpp^{N} \land \tcp^{N,N} \land \beta^{N} \big)
    \lor
    \\ & 
    \big(\tpp^{R} \land \tcp^{R,N} \land \beta^{R} \land (\tpp^{R} \lor \alpha^{R})
    \land (\tpp^{N} \to \alpha^{N}) \big) \lor
    \\ 
    & \bigvee_{K>R}\big(\tpp^{K} \land \tcp^{K,N} \land \beta^{K} \land
    \\ & \hspace*{5em}
    \land (\tpp^{N} \to \alpha^{N}) \land (\tpp^{R} \to \alpha^{R})
    \land \bigwedge_{R<H<K}(\tpp^{H} \to \alpha^{H})
    \big)
    \\
    \equiv\ &
    \zeta^N\Big(\big(\tpp^{J} \land \tcp^{J,J} \land \beta^{J}\big) \lor
    \\ & \quad
    \bigvee_{K>J}\big(\tpp^{K} \land \tcp^{K,J} \land \beta^{K} 
    \land (\tpp^{J} \to \alpha^{J}) \land \bigwedge_{J<H<K}(\tpp^{H} \to \alpha^{H})
    \big)\Big)
    \\
    \equiv\ & \zeta^N(\fPhi{i-1}{\gamma,J}).
  \end{align*}
  \qed
\end{proof}


%% file: app_casestudy.tex
\section{Additional Evaluation Details}
\label{app:casestudy}

\subsection{Compliance Policy Description}
\label{app:policies}

We start by explaining the predicate symbols that model the events
that the banking system is assumed to log or transmit to the monitor.
The predicate $\mathit{trans}(c,t,a)$ represents the execution of the
transaction~$t$ of the customer~$c$ transferring the amount~$a$ of
money.  The predicate $\mathit{report}(t)$ classifies the
transaction~$t$ as suspicious.
Note that a message sent to the monitor describes an event and the
register values. For instance, when executing a transaction, the
registers $\mathit{tid}$ and $\mathit{cid}$ store the identifiers of
the transaction and the customer; the amount of the transaction is
stored in the register $\mathit{sum}$. For a $\mathit{report}$ event,
the register $\mathit{tid}$ stores the identifier of the transaction
whereas the other registers for the customer and the amount store the
default value~$0$.

The formula~(P1) requires that a transaction~$t$ of a customer~$c$
must be reported within at most three time units if the transferred
amount~$a$ exceeds the threshold of \$2,000.
The formulas~(P2) to~(P4) are variants of~(P1).  (P2) requires that
whenever a customer $c$ makes a transaction that exceeds the threshold,
then any of $c$'s future transactions (within a specified period of
time) must eventually be reported (within a specified time bound).
(P3) requires that whenever a customer~$c$ makes a transaction $t$
that exceeds the threshold, then $c$ is not allowed to make further
transactions until the transaction $t$ is reported. Note that the syntactic sugar $\weakuntil$ 
(``weak until'') is used here instead of the primitive
temporal connective $\until$. We not require that the transaction must eventually be reported.
Finally, (P4) requires that whenever a customer $c$ makes a
transaction that exceeds the threshold, then any of $c$'s transactions
in a given time period must be reported.

\subsection{Evaluation in a Propositional Setting}
\label{app:propositonal}

We consider the formulas in Figure~\ref{fig:specifications_prop}
for comparing our experimental results in Section~\ref{sec:casestudy}
with the simpler settings where no data values are involved.  These
formulas are propositional versions of the \MTLdata formula in
Figure~\ref{fig:specifications}, except (P3$'$), which has an
additional temporal connective and accounts for the additional event
$\mathit{unflag}$.
\begin{figure}[t]
  \centering
  \scalebox{.8}{
    \begin{minipage}{15.2cm}
  \begin{gather}
    \tag{P1$'$}
    \always 
    \mathit{transaction} \wedge \mathit{suspicious} \to 
    \eventually_{[0,3]} \mathit{report}
    \\
    \tag{P2$'$}
    \always 
    \mathit{transaction} \wedge \mathit{suspicious} \to
    \always_{(0,3]} \mathit{transaction} \to \neg\mathit{suspicious}
    \\
    \tag{P3$'$}
    \always 
    \mathit{transaction} \wedge \mathit{suspicious} \to
    \big(
      (\mathit{transaction}\rightarrow\eventually_{[0,3]}\mathit{report})
      \weakuntil 
      \mathit{unflag}
    \big)
    \\
    \tag{P4$'$}
    \always 
    \mathit{transaction} \wedge \mathit{suspicious} \to
    \always_{[0,6]}
    \mathit{transaction} \to \eventually_{[0,3]} \mathit{report}
  \end{gather}
    \end{minipage}}
  \caption{MTL formulas.}
  \label{fig:specifications_prop}
\end{figure}

\begin{figure}[t]
  \centering
  \subfigure[in-order]{\includegraphics[scale=.20]{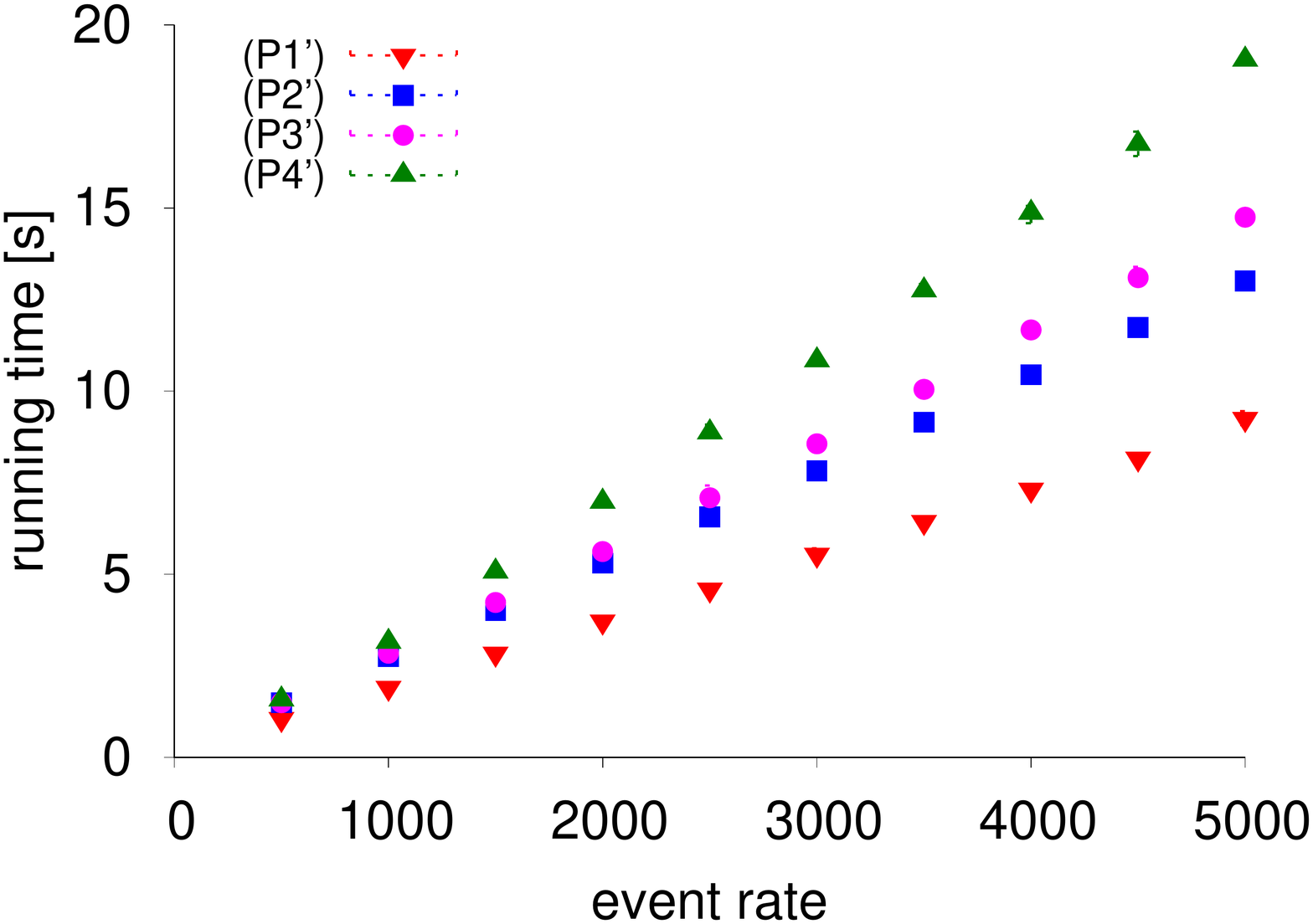}}
  \subfigure[out-of-order]{\includegraphics[scale=.20]{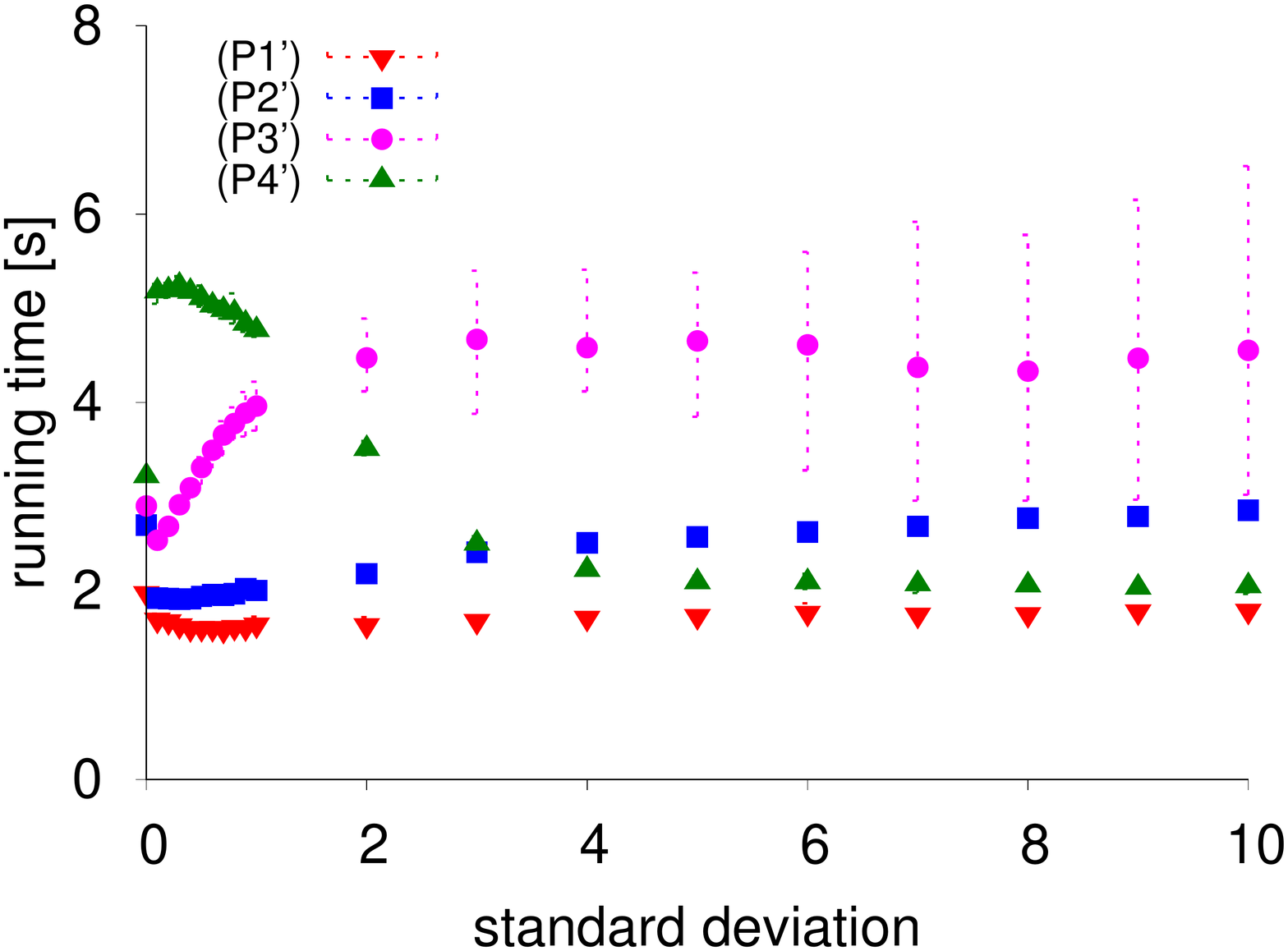}}
  \caption{Running times in a propositional setting (where each data
    point shows the mean of five logs together with the minimum and
    maximum).}
  \label{fig:performance_prop}
\end{figure}
Figure~\ref{fig:performance_prop}(a) shows the running times on logs
with different event rates for the formulas~(P1$'$) to~(P4$'$).
Figure~\ref{fig:performance_prop}(b) shows the impact when messages
are received out of order for logs with an event rate 1000.
We remark that some care must be taken when comparing these figures
with the Figures~\ref{fig:performance}(a) and~(b).
First, the formulas express different policies. For instance, in
(P2$'$) a report might discharge multiple transactions.  Second, the
logs for the propositional settings differ from the logs for the
formulas~(P1) to~(P4). In particular, 
the events in the log files
generated for the propositional settings do not account for different
customers. Furthermore, we have chosen event rates that are 10 times higher.
However, the running times in the propositional setting are
significantly faster. In the propositional setting, our
prototype usually processes an event in a fraction of a millisecond.
